\documentclass[11pt]{article}

\usepackage{fullpage,times,url}
\usepackage{natbib,mathrsfs}

\usepackage{amsthm,amsfonts,amsmath,amssymb,epsfig,color,float,graphicx,verbatim}
\usepackage{algorithm,algorithmic}

\usepackage{enumitem}

\usepackage{hyperref}
\hypersetup{
	colorlinks   = true,
	urlcolor     = blue,
	linkcolor    = blue,
	citecolor   = blue
}

\newtheorem{theorem}{Theorem}[section]
\newtheorem{proposition}[theorem]{Proposition}
\newtheorem{lemma}[theorem]{Lemma}
\newtheorem{corollary}[theorem]{Corollary}
\newtheorem{definition}{Definition}
\newtheorem{remark}[theorem]{Remark}

\newtheorem{lbconstruction}{Lower Bound Construction}

\newcommand{\reals}{\mathbb{R}}
\newcommand{\E}{\mathbb{E}}

\newcommand{\half}{\frac{1}{2}}

\newcommand{\be}{\mathbf{e}}
\newcommand{\bx}{\mathbf{x}}
\newcommand{\bw}{\mathbf{w}}
\newcommand{\bg}{\mathbf{g}}

\newcommand{\bu}{\mathbf{u}}
\newcommand{\bv}{\mathbf{v}}
\newcommand{\bz}{\mathbf{z}}

\newcommand{\by}{\mathbf{y}}

\newcommand{\bs}{\mathbf{s}}
\newcommand{\br}{\mathbf{r}}

\newcommand{\bp}{\mathbf{p}}

\newcommand{\bxi}{\boldsymbol{\xi}}

\newcommand{\Lcal}{\mathcal{L}}
\newcommand{\Ocal}{\mathcal{O}}
\newcommand{\Acal}{\mathcal{A}}

\newcommand{\Ncal}{\mathcal{N}}

\newcommand{\Wcal}{\mathcal{W}}

\newcommand{\norm}[1]{\|#1\|}

\newcommand{\secref}[1]{Section~\ref{#1}}

\renewcommand{\eqref}[1]{Eq.~(\ref{#1})}
\newcommand{\lemref}[1]{Lemma~\ref{#1}}

\newcommand{\thmref}[1]{Theorem~\ref{#1}}
\newcommand{\propref}[1]{Proposition~\ref{#1}}

\title{
The Oracle Complexity of Simplex-based Matrix Games
}
\author{
Guy Kornowski$^{1}$
\qquad
Ohad Shamir$^{1,2}$
\vspace*{2pt}
\\
$^{1}$Weizmann Institute of Science
\\
$^{2}$University of Toronto
\\
\texttt{\{guy.kornowski,ohad.shamir\}@weizmann.ac.il}
}
\date{}

\begin{document}

\maketitle

\begin{abstract}
We study the problem of solving matrix games of the form $\min_{\mathbf{p}\in\Delta}\max_{\mathbf{w}\in\mathcal{W}}\mathbf{p}^{\top}A\mathbf{w}$, where $A$ is a matrix and $\Delta$ is the probability simplex. This problem encapsulates canonical tasks such as finding a linear separator and computing Nash equilibria in zero-sum games. However, perhaps surprisingly, its inherent complexity (as formalized in the standard framework of oracle complexity) is not well understood. In this work, we first identify different oracle models that are implicitly used by prior algorithms, corresponding to multiplying the matrix $A$ by a vector from either one or both sides. We then prove complexity lower bounds for algorithms under both access models, which in particular imply a separation between them.
As our main result, we prove that in the general 
$\ell_p$/simplex setting where $\mathcal{W}$ is an $\ell_p$ ball for $p\in[1,\infty]$, any algorithm  that utilizes two-sided matrix-vector multiplications requires $\tilde{\Omega}(\epsilon^{-2/3})$ iterations to return an $\epsilon$-suboptimal solution. For any $p\in[1,\infty]$, this is either the first lower bound for such problems, or an exponential improvement over the previously best-known results.
Moreover, for the canonical tasks of finding a linear separator and computing a Nash equilibrium, our lower bounds match (up to log factors) recent algorithms of Karmarkar, O’Carroll and Sidford (2026), thereby resolving their oracle complexities in a natural setting.
\end{abstract}

\section{Introduction}

Given a matrix $A\in\reals^{n\times d}$ and some convex domain $\Wcal\subset\reals^d$,
we consider the
bilinear minimax optimization problem (also known as a matrix game)
\begin{equation} \label{eq: problem}
\min_{\bp\in\Delta^{n-1}}\max_{\bw\in\Wcal}\bp^\top A\bw 
~=~
\max_{\bw\in\Wcal}\min_{\bp\in\Delta^{n-1}}\bp^\top A\bw ~=~\max_{\bw\in\Wcal}~\min_{l\in \{1,\ldots,n\}} (A\bw)_{l}~,
\tag{A}
\end{equation}
where $\Delta^{n-1}:=\{\bp\in\reals^n: \min_i p_i\geq 0, \sum_{i=1}^{n}p_i=1\}$ is the probability simplex. We denote the optimal value of this problem as $\gamma_A$. 
The problem of finding $\bw\in\Wcal$ that approximates the optimum of such problems
is extensively studied throughout machine learning, statistics, optimization and economics, as several important problems take this form, depending on the choice of the set $\Wcal$.
For the sake of introduction, we discuss two prominent cases:

\begin{enumerate}
\item \textbf{Linear separability:} When $\Wcal=\{\bw\in\reals^d:\|\bw\|_2\leq 1\}$ is the unit Euclidean ball, Problem (\ref{eq: problem}) corresponds to the canonical problem of finding a linear separator (namely, a vector $\bw$ such that $A_l\bw>0$ for all rows $A_l$ of $A$), and the optimal value $\gamma_A$ is known as the \emph{margin}.
This is a fundamental classification and linear programming
problem which can be dated back to the work of \citet{mcculloch1943logical},
and is solved for instance
by the well-known Perceptron algorithm
\citep{rosenblatt1958perceptron}.

\item \textbf{Nash equilibria:}
When $\Wcal=\Delta^{d-1}\subset\reals^d$ is the simplex,
Problem (\ref{eq: problem}) corresponds 
to
maximizing the utility of a player in a zero-sum game,
and $\gamma_A$ is known as the game's \emph{value}.
Due to the minimax theorem, by symmetrically solving this problem for each of the players (namely, $\bp$ and $\bw$), this objective is equivalent
to the canonical problem of
finding a Nash equilibrium, or saddle point, in zero-sum matrix games
\citep{nash1950equilibrium}.

\end{enumerate}

Perhaps surprisingly, although these are canonical problems with a long history (see \secref{sec: oracle models}), the inherent complexity of solving them is relatively little studied. The goal of this paper is to study this question through the standard framework of oracle complexity \citep{nemirovskiyudin1983}. Specifically, we are interested in the performance limits of iterative algorithms, where each iteration is based on a simple computation involving the matrix $A$, such as multiplying it by some vector or extracting a row of $A$. This interaction between the algorithm and the matrix can be modeled as accessing an oracle, which simulates this computation and provides the algorithm with the result. We can then ask and rigorously study how many such oracle queries/computations are required, so that an algorithm with no prior knowledge of $A$ will solve the associated matrix problem up to a given level of precision (see \secref{sec:preliminaries} for more details). In this paper, we focus on the high-dimensional setting, where the size of the matrix $A$ is essentially unrestricted, and we are interested in bounds which are independent of (or only weakly dependent on) the matrix size. 

Considering linear separability, it is well known that for matrices with normalized rows, the Perceptron algorithm finds a linear separator using $O(\gamma_A^{-2})$ iterations (each involving a single matrix-vector multiplication and extraction of a single row of $A$), independently of $n,d$ \citep{novikoff1962convergence}.
Moreover, the same iteration complexity
is achievable in non-Euclidean settings whenever $\Wcal$ is an $\ell_p$ ball with $p\in[1,2]$ (with each iteration involving the same access to $A$) \citep{gentile1999robustness,gentile2001new}.
Half a century later, \citet{soheili2012smooth} and \citet{yu2014saddle}
used acceleration techniques due to \citet{nesterov2005excessive} and \citet{nemirovski2004prox} respectively, and
provided accelerated algorithms which find a separator using only $O(\sqrt{\log n}\cdot\gamma_A^{-1})$ iterations. However, a closer inspection of these methods reveal that they rely on a stronger oracle access to the matrix $A$, amounting to 
multiplying $A$ by vectors \emph{on both sides}
(instead of just one-sided multiplications as required by the Perceptron algorithm and its non-Euclidean variants).
Thus, one may ask whether such two-sided operations are necessary for these accelerated results, and how close they are to being optimal.

As for approximating Nash equilibria,
there are several algorithms in the literature which are known to return an $\epsilon$-suboptimal solution\footnote{More precisely, these algorithms guarantee returning an $\epsilon$-approximate saddle point, namely $(\hat{\bw},\hat{\bp})\in\Delta^{d-1}\times\Delta^{n-1}$ such that
$\max_{\bw\in\Delta^{d-1}}\hat{\bp}^\top A\bw-\min_{\bp\in\Delta^{n-1}}\bp^\top A \hat{\bw}\leq\epsilon$. This implies $\epsilon$-suboptimality, since $\max\{
\max_{\bw\in\Delta^{d-1}}\hat{\bp}^\top A\bw-
\gamma_A,~
\gamma_A-\min_{\bp\in\Delta^{n-1}}\bp^\top A\hat{\bw}
\}\leq\epsilon$ where $\gamma_A$ is the optimal value.
Since the latter inequality also implies the former (with $2\epsilon$ replacing $\epsilon$), we see that the two notions are the same up to a factor of 2.
}
using $O(\sqrt{\log(n)\log (d)}\cdot\epsilon^{-1})$ two-sided matrix multiplication queries \citep{nemirovski2004prox,nesterov2005smooth,nesterov2007dual,rakhlin2013optimization}.
Corresponding lower bounds are scarce, and only recently \citet{hadiji2024towards} proved a lower bound of $\Omega(\log(1/n\epsilon))$ for sufficiently small $\epsilon=\mathrm{poly}(1/n)$ and $n=d$.

In a very recent line of work, Karmarkar, O'Carroll and Sidford \citep{karmarkar2025solving,karmarkar2026solving}
improved the best-known complexities for these two canonical tasks, proving that in fact $\tilde{O}(\epsilon^{-2/3})$ two-sided matrix multiplication queries
suffice for returning a linear separator as well as an $\epsilon$-approximate Nash equilibrium, which constitutes a significant polynomial improvement in both settings.
Nonetheless, previous to this work, the optimal possible rates for solving these canonical tasks remained unknown.

\paragraph{Our contributions.}

In this work, we study the oracle complexity of solving matrix games involving the simplex
as in Problem (\ref{eq: problem}). As mentioned earlier, we focus on the high-dimensional regime, where the bounds should be independent (or at least not polynomial) in the matrix size $n,d$, and the matrix $A$ satisfies suitable magnitude constraints.
Our contributions can be summarized as follows.

\begin{itemize}
\item \textbf{One-sided vs. two-sided oracles (\secref{sec: oracle models}:)}
We start by identifying and formalizing different oracle models
for matrix games that are implicitly used by existing algorithms, as we will see that these 
oracles lead to different complexities.
One oracle model corresponds to
querying rows of $A$, together with ``one-sided'' multiplication queries of the form $\bw\mapsto A\bw$. The second (and stronger) oracle model we will consider allows ``two-sided'' multiplications $(\bp,\bw)\mapsto(\bp^\top A,A\bw)$.

    \item \textbf{
    Linear separability with a one-sided oracle (Theorem~\ref{thm:one-sided}):}
    We first show that any deterministic algorithm that performs row queries and one-sided multiplication queries, must require $\Omega(\gamma_A^{-2})$ queries in the worst-case in order to find a linear separator. More generally, we show this lower bounds holds (up to log factors) whenever $\Wcal$ is an $\ell_p$ ball with $p\in[1,2]$.
    The claim is proved using a classic lower bound technique due to \citet{nemirovskiyudin1983}.
    In particular, this establishes the optimality of the Perceptron algorithm under this oracle model, as well as its non-Euclidean variants.

\item \textbf{General simplex/$\ell_p$ games with a two-sided oracle (Theorem~\ref{thm:main}):} 
We prove that whenever the domain $\Wcal$ is an $\ell_p$ ball for $p\in[1,\infty]$,
any deterministic algorithm which performs two-sided multiplication queries must require $\tilde{\Omega}(\epsilon^{-2/3})$ queries in the worst-case in order to find an $\epsilon$-suboptimal solution for Problem~(\ref{eq: problem}). No such lower bound was knowbefore this work for \emph{any} value of $p$, and to the best of our knowledge, this is the first polynomial oracle complexity lower bound for either linear separability or solving zero-sum games, which applies even to accelerated algorithms. As opposed to the lower bound for the one-sided oracle, the proof of the lower bound for the two-sided oracle is substantially more involved, and requires some new proof ideas. 
In particular, our result carries several notable implications:

\begin{itemize}
    \item For $p=2$, our result implies that $\tilde{\Omega}(\gamma_A^{-2/3})$ matrix-vector queries are required to find a linear separator for data with margin $\gamma_A$.
    Together with \citet{karmarkar2026solving}, our work resolves the oracle complexity of this task (up to log factors) in the natural model of matrix-vector multiplications.

    \item For $p=1$, via a reduction from the simplex to the $\ell_1$ ball, our result implies that $\tilde{\Omega}(\epsilon^{-2/3})$ matrix-vector queries are required to compute an $\epsilon$-suboptimal strategy in zero-sum games, and hence an $\epsilon$-Nash equilibrium. This is an exponential improvement over the previously known lower bound by \citet{hadiji2024towards}.  Together with \citet{karmarkar2026solving}, our work resolves the oracle complexity of this task (up to log factors) in the natural model of matrix-vector multiplications.
\end{itemize}

\end{itemize}

We further discuss implications of our results to other non-Euclidean settings that appear in the literature, and conclude with a discussion in \secref{sec:discussion}.

\section{Preliminaries}\label{sec:preliminaries}

\paragraph{Notation.}
We use capital letters to denote matrices, and bold-face letters to denote vectors. $[n]$ is shorthand for $\{1,\ldots,n\}$.
Given $p\in[1,\infty]$, we denote by $q$ its dual exponent, i.e., $1/p+1/q=1$
and use the convention $1/\infty=0$.
Vectors are always in column form. Given an indexed vector $\bv_t$, $v_{t,i}$ denotes its $i$-th entry. $\be_i$ denotes the $i$-th standard basis vector. $\mathbf{1}$ is the all-ones vector. Given a matrix $A$, $A_l$ refers to its $l$-th row. $\norm{\cdot}$ refers to the Euclidean norm $\norm{\cdot}_2$, $\norm{\cdot}_p$ refers to the $\ell_p$ norm, and $\|A\|_{p\to\infty}:=\max_{\|\bw\|_p\leq 1}\|A\bw\|_\infty=\max_{l\in[n]}\|A_l\|_q$ (the latter holds by H\"older's inequality).
$\log(\cdot)$ refers to the natural logarithm.

\paragraph{Oracle Complexity.}
As described in the introduction, we study the complexity of matrix games via the well-known optimization-theoretic frameork of oracle complexity \citep{nemirovskiyudin1983}. This framework focuses on iterative methods, where each iteration utilizes some restricted information about the
relevant
objective function. In our context of matrix games as in (\ref{eq: problem}), we assume that the domain $\Wcal$ is fixed and known, and that the matrix $A$ is known to belong to some set $\mathscr{A}$, which captures a magnitude constraint for the optimization problem.
Specifically, whenever $\Wcal$ is an $\ell_p$ ball,
we require $\|A\|_{p\to\infty}\leq 1$, which is equivalent to ensuring that $|A_l \bw|\leq 1$ for any row $l\in[n]$ and $\bw$ in the unit $\ell_p$ ball.
For linear separability for example, this correspond to all $n\times d$ matrices whose rows have Euclidean norm at most $1$; similarly, for Nash equilibrium, it is the set of all $n\times d$ matrices whose entries have values in $[-1,+1]$.
Crucially, the algorithm has no additional prior knowledge of $A$. In order to solve the matrix problem, the algorithm has access to an oracle $\Ocal(\cdot)$ which provides some limited information about $A$: For example, given the vectors $\bp\in\reals^n, \bw\in \reals^d$ chosen by the algorithm, the oracle returns $\bp^\top A$ and $A\bw$. The algorithm interacts with the oracle for a given number of iterations, after which it returns an output as a function of all previous oracle responses. One can then ask how many iterations are required by some algorithm, for the output to satisfy a certain performance metric for all $A\in\mathscr{A}$ (as a function of the problem parameters $n,d,\mathscr{A},\Wcal$ and the type of oracle $\Ocal(\cdot)$). This framework naturally models standard scalable approaches for solving matrix games, and allows one to prove both upper bounds and unconditional lower bounds for iterative algorithms, assuming they interact with the matrix $A$ in a manner corresponding to the given oracle. 

\begin{remark}
In this paper, we focus on deterministic algorithms, whose oracle queries and output are deterministic function of the previous oracle responses. Since all state-of-the-art algorithms for the problems we consider are deterministic, this is without too much loss of generality. However, extending our lower bounds to randomized algorithms is certainly an interesting direction for future work (see \secref{sec:discussion} for more details). 
\end{remark}

\section{Oracle models} \label{sec: oracle models}

To motivate the oracles that we consider, let us begin by examining the classical Perceptron algorithm \citep{rosenblatt1958perceptron} for linear separability: Given the matrix $A$ (and assuming that a linear separator exists), the algorithm iteratively searches for a row $A_l$ of $A$ such that $A_l\bw<0$ (where $\bw$ is the current iterate), and then adds $A_l$ to $\bw$. This process is repeated until no such rows are found. It is well-known that this algorithm will terminate in at most $O(\gamma_A^{-2})$ iterations, resulting in a linear separator $\bw$ such that $\min_{l\in [n]}(A\bw)_l > 0$. From an oracle complexity perspective, each iteration of the algorithm can be modeled via two operations on $A$: One is a right matrix-vector multiplication $\bw\mapsto A\bw$, and the second is the extraction of a row of $A$ whose inner product with $\bw$ is negative. We can formally model these operations via the following oracle: 
\begin{definition}[One-sided Oracle $\Ocal_1^{A}$]
    Given some $\bw\in\reals^d$ and index $l\in [n]$, the oracle $\Ocal_1^{A}(l,\bw)$ returns $A\bw$ and $A_{l}$. 
\end{definition}
Thus, the Perceptron's convergence guarantee implies that $O(\gamma_A^{-2})$ queries to a one-sided oracle is sufficient for finding a linear separator.
In fact, a more general result can be achieved by applying the well-known subgradient method on  the equivalent convex problem $\min_{\bw:\norm{\bw}\leq 1}\max_{\bp\in\Delta}\bp^\top (-A)\bw$. Specifically, by standard results, this method requires $O(1/\epsilon^2)$ subgradient computations in order to find a vector $\bw$ whose value is $\epsilon$-suboptimal \citep{nemirovskiyudin1983,nesterov2018lectures}. This holds for any $\epsilon$, whereas the guarantee for the Perceptron is merely for the special case $\epsilon=\gamma_A$ (namely, we seek a solution whose value is $>0$, with the optimal value being $\gamma_A$). From an oracle-complexity perspective, implementing such methods requires access to supergradients of the function $f(\bw)=\min_{\bp\in\Delta^{n-1}}\bp^\top A\bw = \min_{l\in[n]} (A\bw)_l$, which equal $A_{l_{\min}}$ where $l_{\min}\in\arg\min_{l\in [n]}(A\bw)_{l}$. Thus, we can model these methods as iteratively interacting with the following supergradient oracle:
\begin{definition}[Supergradient Oracle $\Ocal_\partial^{A}$]
    Given some $\bw\in\reals^d$, the oracle $\Ocal_\partial^{A}(\bw)$ returns $A_{l_{\min}}$ where $l_{\min}\in\arg\min_{l\in [n]}(A\bw)_{l}$.
\end{definition}
Note that this oracle is strictly weaker than a one-sided oracle (up to a factor of $2$): On the one hand, we can simulate each call to a supergradient oracle by two calls to a one-sided oracle. On the other hand, a one-sided oracle allows us to extract any single row of $A$, and not just the one corresponding to the smallest entry in $A\bw$. In what follows, we will prove lower bounds for any algorithm based on a one-sided oracle, and thus the lower bounds automatically extend to any algorithm based on a supergradient oracle. 

As we discussed earlier, works such as \cite{soheili2012smooth} and \cite{yu2014saddle} show that the $O(\gamma_A^{-2})$ iteration bound of the Perceptron algorithm can actually be improved. In a nutshell, this is achieved by applying accelerated gradient methods on top of a smoothing of the objective function (using, for example, the log-sum-exp function instead of a hard max). These result in bounds of the form $O(\sqrt{\log(n)}/\gamma_A)$ for matrices with margin parameter $\gamma_A$, or more generally, $O(\sqrt{\log(n)}/\epsilon)$ iterations to get an $\epsilon$-optimal solution. Accelerated methods to maximize the margin were also proposed in \citep{ji2021fast,wang2023accelerated}.

Why can these accelerated methods beat the Perceptron bound, from an oracle-complexity perspective? A close inspection of these methods reveal that they all actually require a stronger oracle than a supergradient (or even one-sided) oracle: They crucially require \emph{two-sided} matrix-vector multiplications. In more detail, accelerated gradient methods can optimize convex functions with $L$-Lipschitz gradients to suboptimality $\epsilon$ with $O(\sqrt{L/\epsilon})$ gradient computations. Moreover, for any $\epsilon$, the $\min$ (or $\max$) operator can be approximated to accuracy $\epsilon$ using a smooth function $\tilde{f}$ with $(\log(n)/\epsilon)$-Lipschitz gradients. Combining these two observations, it follows that one can optimize the original matrix problem to accuracy $2\epsilon$, using $O(\sqrt{(\log(n)/\epsilon)/\epsilon})=O(\sqrt{\log(n)}/\epsilon)$ gradient computations of the function $\bw\mapsto \tilde{f}(A\bw)$. The gradient of this function is given by $A^\top\tilde{f}'(A\bw)$, so its computation requires multiplying the matrix $A$ from both the left and from the right. Thus, we are led to the following natural \emph{two-sided} oracle model: 
\begin{definition}[Two-sided Oracle $\Ocal_2^{A}$]
    Given some $\bp\in\reals^n,\bw\in\reals^d$, the oracle $\Ocal_2^A(\bp,\bw)$ returns $A\bw$ and $\bp^\top A$.
\end{definition}
Clearly, a two-sided oracle is stronger than
a
one-sided
oracle,
since $\Ocal_1^{A}(l,\bw)$ can be simulated by $\Ocal_2^{A}(\be_l,\bw)$. As far as we know, all existing accelerated algorithms for linear separability can be implemented with such an oracle, so any lower bound for algorithms based on this oracle will apply to them.
We also note that lower bounds with respect to two-sided oracles were studied in the context of solving linear equations \citep{nemirovsky1992information}, which does not directly relate to our results or proofs.

Although we have considered so far the linear separability problem, a two-sided oracle is also very natural to model algorithms for other matrix games. In particular, for computing a Nash equilibrium, a two-sided oracle corresponds precisely to a first-order (or gradient) oracle, which given $\bw,\bp$, returns the gradient of the function $(\bw,\bp)\mapsto \bp^\top A\bw$ (namely $\bp^\top A$ and $A\bw$). Under this well-studied setting, the best-known method dating back to \citet{nemirovski2004prox} requires $O(\sqrt{\log(n)\log(d)}/\epsilon)$ two-sided oracle calls in order to find an $\epsilon$-optimal solution (see also \citet{nesterov2005smooth,nesterov2007dual,rakhlin2013optimization} and \citet{hadiji2024towards} for a detailed discussion of other results). However, very little work appears to exist on lower bounds, with the notable exception of \cite{hadiji2024towards}, whose bound is $\Omega(\log(1/(n\epsilon)))$ -- namely, logarithmic in $1/\epsilon$ --  and when $n=d$. 

\begin{remark}
In all oracle definitions, we do not restrict the input $\bw$ to lie in $\Wcal$, nor $\bp$ to lie in $\Delta^{n-1}$. Thus, our lower bounds will apply equally to algorithms which can query outside these domains.
\end{remark}

\subsection*{Additional Related work} 

\paragraph{Dimension-dependent oracle complexity.}
Besides the accelerated algorithms for linear separability discussed earlier (with convergence rate $O(\log(n)\cdot\epsilon^{-1})$), there exists another family of rescaling-based methods, which can achieve an even faster (logarithmic) dependence on $\epsilon^{-1}$, but with iteration bounds scaling polynomially with the matrix dimensions \citep{dunagan2004polynomial,belloni2009efficient,pena2016deterministic,dadush2020rescaling}. This is akin to the situation in convex optimization, where one can trade off between algorithms with dimension-independent, $\text{poly}(\epsilon^{-1})$-dependent guarantees (using gradient or subgradient methods), and algorithms with poly-dimension-dependent, $\log(\epsilon^{-1})$-dependent guarantees (using methods such as interior points, ellipsoids, or center-of-gravity, see \citealp{nesterov2018lectures}).
Similarly, several works developed algorithms for general matrix games that reduce the $\epsilon$-dependence at the cost
of polynomial dependencies on $n,d$, thus improving size-independent algorithms in certain parameter regimes \citep{carmon2019variance,carmon2024whole}.
Since our focus is on size-independent (or near independent) bounds, these family of methods are orthogonal to our work, although understanding the ultimate limits in those regimes is interesting as well.

\paragraph{Other oracle models.}
Going beyond matrix games, there exist quite a few oracle complexity lower bounds for more general minmax convex-concave optimization problems (e.g., \citealp{ibrahim2020linear,ouyang2021lower}), but the constructions do not apply to matrix games as we consider. \citet{carmon2020acceleration} considered the oracle complexity of optimizing the maximum of several linear functions using a ball oracle, which returns the optimum in a small ball around a given point. Although the structure of the objective function is closely related to ours, the oracle is different than the one we consider here (as far as we can surmise). Moreover, their lower bound construction requires non-homogeneous linear functions, which are not included in the matrix game settings that we consider. \citet{clarkson2012sublinear} provide an $\Omega(\gamma_A^{-2})$ oracle complexity lower bound for linear separability, but for a weaker oracle which only returns individual matrix entries.

\section{Warm up: Linear separability with a one-sided oracle}

In this section, we assume that $\Wcal$ is the unit $\ell_p$ ball in $\reals^d$ for some $p\in[1,2]$, so the problem of interest is 
\begin{equation}\label{eq:l2}
\max_{\bw\in \reals^d:\norm{\bw}_p\leq 1}~\min_{\bp\in\Delta^{n-1}}\bp^\top A\bw
~=~\max_{\bw\in \reals^d:\norm{\bw}_p\leq 1}~\min_{l\in [n]} (A\bw)_{l}~.
\end{equation}
As previously discussed, the special case $p=2$ corresponds to the canonical problem of finding a max-margin linear separator (in the Euclidean norm) for a dataset comprised of $A$'s rows. More generally, non-Euclidean generalization of the Perceptron algorithm that apply for any $p\in[1,2]$ have all been shown to require $O(\gamma_A^{-2})$ one-sided oracle queries for any $p\in[1,2]$ \citep{gentile1999robustness,gentile2001new}.

We begin by formally showing that any algorithm using a one-sided oracle for $T$ iterations cannot find a linear separator (a vector $\bw$ such that $A\bw$ has only positive entries), if the margin parameter is less than $\tilde{\Omega}(1/\sqrt{T})$. Since a supergradient oracle is weaker than a one-sided oracle, the same result automatically applies to any algorithm based on a supergradient oracle.

\begin{theorem}\label{thm:one-sided}
   Suppose $d>4T+32\log(4T)$. Then for any deterministic algorithm for solving Problem (\ref{eq:l2}) for some $p\in[1,2]$, there exists a $(T+1)\times d$ matrix $A$ satisfying $\|A\|_{p\to\infty}=\max_{l\in [T+1]}\norm{A_{l}}_q\leq 1$ and
    \[
   \max_{\bw:\norm{\bw}_p\leq 1}\min_{\bp\in \Delta^{n-1}}\bp^\top A\bw~\geq~
   \frac{1}{\sqrt{T+1}}
 \cdot\begin{cases}
     1 & \text{if~~}p=2
     \\
     \frac{1}{64\log^{3/2}(4d)}& \text{if~~}p\in[1,2)
 \end{cases}~,
    \]
    yet after $T$ rounds of interaction with a one-sided oracle $\Ocal_1^A$, the algorithm returns a vector $\bw_{T+1}$ such that
    \[
    \min_{\bp\in \Delta^{n-1}}\bp^\top A \bw_{T+1}~\leq~0~.
    \]
\end{theorem}
We note that in the theorem, we require $\|A\|_{p\to\infty}\leq 1$, to ensure that $|A_l \bw|\leq 1$ for any row $l\in[n]$ and $\bw$ in the unit $\ell_p$ ball, which is the relevant normalization of $A$ in this problem (of course, one can also re-scale $\|A\|_{p\to\infty}$ to any other fixed value without loss of generality).

By fixing some $\gamma_A>0$ and setting
$T=\lfloor2^{-12}\gamma_A^{-2}\log^{-3}(4d)\rfloor,~d=\lceil4T+32\log(4T) \rceil$, we can restate this result as follows: For any $\gamma_A$, and for any algorithm whose interaction with $A$ is captured by a one-sided oracle, there exists a matrix $A$ (with unit rows in $q$-norm and margin parameter at least $\gamma_A$) such that the required number of iterations to find a linear separator is $\tilde{\Omega}(\gamma_A^{-2})$,
and even $\Omega(\gamma_A^{-2})$ when $p=2$, as claimed in the introduction.
This reassures the optimality of the Perceptron algorithm in this model of computation, and more generally, of its non-Euclidean generalizations (up to log factors).

We emphasize that the proof of this particular theorem is a rather straightforward adaptation of existing techniques, and its main purpose is to complete the picture regarding the power of different oracles to solve this matrix game. Nevertheless, the proof below illustrates the iterative nature of such constructions, concretely via iterative Gaussian projections, an idea that also appears in the proofs of our other results to follow.

\begin{remark}
It should be possible to remove the log factor $1/\log^{3/2}(d)$ that appears in the stated result for $p\in[1,2)$, at the cost of a more involved analysis.
The main purpose of our presentation here is to set the stage for the two-sided results to follow,
and to prove a \emph{polynomial} separation of the rates of one- vs. two-sided oracles.
\end{remark}

\begin{proof}[Proof of \thmref{thm:one-sided}]
    The proof is directly inspired by standard oracle complexity lower bounds for convex Lipschitz optimization due to \citet{nemirovskiyudin1983}, as well as the standard mistake lower bound proof of the Perceptron algorithm (cf. \citealp[Section 9.5, Question 3]{shalev2014understanding}). The main idea is to construct $A$'s row as mutually orthogonal unit vectors, which are also orthogonal to the algorithm's queries $\bw_t$. Therefore, the algorithm can recover the rows one at a time, by querying for a particular row. However, if the number of rows is larger than $T$, there will remain rows orthogonal to $\bw_{T+1}$. Therefore, $\bw_{T+1}$ will not be a linear separator for $A$. On the other hand, since $A$ has $T+1$ orthogonal rows, it can be shown to be linearly separable with margin roughly $1/\sqrt{T}$.

More formally, given an algorithm $\Acal$, consider the following iterative construction:
    \begin{itemize}
    \item Initialize $\Lcal_0=\emptyset$, and $A_0$ to be the the all-zeros $(T+1)\times d$ matrix.
    \item For $t=1,2,\ldots,T:$
    \begin{itemize}
        \item Compute algorithm queries $\bw_t,l_t$ (based on oracle outputs received so far).
        \item Set $A_{t}:=A_{t-1}$ and $\Lcal_t:=\Lcal_{t-1}$. 
        \item If $l_t\notin \Lcal_{t-1}$
        \begin{itemize}
            \item Let $(A_{t})_{l_t}=\frac{1}{\|\Pi_t(\bxi_t)\|_q}\Pi_t(\bxi_t)$ be the normalized projection of an independent Gaussian $\bxi_t\sim\Ncal(\mathbf{0},I_d)$, where $\Pi_t:\reals^d\to\mathrm{span}((\bw_i)_{i=1}^{t},(A_t)_{\Lcal_{t}})^\perp$ is the projection onto the subspace orthogonal to $\bw_1,\ldots,\bw_{t}$ as well the rows of $A_t$ indexed by $\Lcal_t$.
            \item Set $\Lcal_t := \Lcal_{t-1}\cup \{l_t\}$.
        \end{itemize}
        \item Feed $\Acal$ with $A_t\bw_t$ and $(A_{t})_{l_t}$ (as a response for its queries $\bw_t,l_t$). 
    \end{itemize}
    \item Compute algorithm output $\bw_{T+1}$ (based on oracle outputs so far).
    \item Set $A=A_{T}$, and set all rows $l\notin \Lcal_{T}$ of $A$ to be some mutually orthogonal vectors which are also orthogonal to $\bw_1,\ldots,\bw_{T+1}$ with unit $q$-norm.
    \end{itemize}
    We note that the dimensionality $d$ is sufficiently high to ensure that the rows of $A$ are mutually orthogonal unit vectors as specified.
    Moreover, it is easy to verify that because of the orthogonality, it holds that $A_t\bw_t=A\bw_t$ for all $t\in [T]$: Namely, the responses given to the algorithm are consistent with the oracle responses on the matrix $A$. However, after $T$ iterations, $|\Lcal_{T}|\leq T$, yet there are $T+1$ rows. Therefore, at least one row of $A$ will be chosen to be orthogonal to $\bw_{1},\ldots,\bw_{T+1}$, and in particular to $\bw_{T+1}$. Hence, $A\bw_{T+1}$ contains a $0$ entry, so
    \[
    \min_{\bp\in\Delta^{n-1}}\bp^\top A\bw_{T+1}~\leq~0~.
    \]

    On the other hand, the vector $\bw:=\frac{1}{\|\sum_{t=1}^{T+1}A_{t}\|_p}\sum_{t=1}^{T+1}A_{t}$ has unit $\ell_p$ norm,
    and satisfies by the orthogonality property for any $l\in[T+1]:$
    \begin{align} \label{eq: Awl}
(A\bw)_l
= \frac{1}{\|\sum_{t=1}^{T+1}A_{t}\|_p}\sum_{t=1}^{T}A_t^\top A_l
= \frac{\|A_l\|_2^2}{\|\sum_{t=1}^{T+1}A_{t}\|_p}~.
    \end{align}
Recalling that $A$'s rows are projected Gaussians, and therefore are Gaussians themselves,\footnote{To be precise, they are Gaussians with random covariance matrices.} we can invoke Gaussian concentration arguments to lower bound the ratio above. This results in the following lemma, whose proof is deferred to Section~\ref{sec: proofs}.

\begin{lemma} \label{lem: one sided}
$A$ constructed as above satisfies with probability at least $\half$ that for all $l\in[T+1]:$
\[
\frac{\|A_l\|_2^2}{\|\sum_{t=1}^{T+1}A_{t}\|_p}
 \geq
\frac{1}{\sqrt{T+1}}
 \cdot\begin{cases}
     1 & \text{if~~}p=2
     \\
     \frac{1}{64\log^{3/2}(4d)}& \text{if~~}p\in[1,2)
 \end{cases}
~.
\]
\end{lemma}

Overall, by the probabilistic method, there is a suitable matrix $A$ so that \eqref{eq: Awl} and \lemref{lem: one sided} ensure that
    \[
    \max_{\bw:\norm{\bw}_p\leq 1}\min_{\bp\in\Delta^{n-1}}\bp^\top A\bw
    ~\geq~
     \frac{1}{\sqrt{T+1}}
 \cdot\begin{cases}
     1 & \text{if~~}p=2
     \\
     \frac{1}{64\log^{3/2}(4d)}& \text{if~~}p\in[1,2)
 \end{cases}~.
    \]

\end{proof}

\section{Simplex/$\ell_p$ games with a two-sided oracle}

We now turn to an oracle complexity lower bound for the general class of simplex/$\ell_p$ matrix games for some $p\in[1,\infty]$, namely problems of the form
\begin{equation} \label{eq: problem lp}
\max_{\bw:\|\bw\|_p\leq 1}\min_{\bp\in\Delta^{n-1}}\bp^\top A\bw ~,
\end{equation}
with respect to the much stronger two-sided oracle $\Ocal_2^A$.

\begin{theorem}\label{thm:main}
The following holds for some large enough universal constant $C>0$ and any $p\in[1,\infty]$: For any $T>C$, suppose $n,d$ are sufficiently large so that $d>CT$ and $n>CT^2\log(T)$.
    Then for any deterministic algorithm for solving Problem~(\ref{eq: problem lp}), there exists a matrix $A\in \reals^{n\times d}$ satisfying $\|A\|_{p\to\infty}\leq 1$ and
\begin{align*}
\max_{\bw:\norm{\bw}_p\leq 1}\min_{\bp\in \Delta^{n-1}}\bp^\top A\bw~&\geq~
\frac{1}{CT(\sqrt{T\log(nd)}+\log(nd))}
\cdot 
\begin{cases}
    1 & \text{if~~}p\in[1,2]
    \\
     1/\sqrt{\log(d)} & \text{if~~}p\in(2,\infty]
    \end{cases}
~,
\end{align*}
yet after $T$ rounds of interaction with the two-sided oracle $\Ocal_2^A$, the algorithm returns a vector $\bw_{T+1}$ such that
    \[
    \min_{\bp\in \Delta^{n-1}}\bp^\top A \bw_{T+1}~\leq~0~.
    \]   
\end{theorem}

\begin{remark}
Note that by setting $n,d$ in the result above so that $\sqrt{T\log(nd)}\gg\log(nd)$, the lower bound scales as $\Omega(1/T^{3/2}\sqrt{\log(nd)})$ for $p\in[1,2]$, and as $\Omega(1/T^{3/2}\log(nd))$ for $p\in(2,\infty]$. In any case, a crude bound that holds regardless of the values of $n,d$ and $p$ is $\Omega(1/T^{3/2}\log^{3/2}(nd))$.
\end{remark}

As before, we can state this lower bound in terms of the suboptimalty. By fixing small enough $\epsilon>0$ and choosing $T,n,d$ such that the lower bound in the theorem is at least $\epsilon$, we can restate the bound as follows: If the number of iterations $T$ is less than $\Omega(\log^{-1/3}(nd)\epsilon^{-2/3})$ when $p\in[1,2]$ or less than $\Omega(\log^{-2/3}(nd)\epsilon^{-2/3})$ if $p\in(2,\infty]$, there exists a $n\times d$ matrix $A$ such that the optimal value is at least $\epsilon$, yet the algorithm's output has value less than $0$, hence is suboptimal by at least $\epsilon$.

We now turn to discuss the proof of \thmref{thm:main}. We begin by noting that the proof technique of \thmref{thm:one-sided} is of no use here, as it is based on revealing the rows of $A$ to the algorithm one at a time. This is not possible with a two-sided oracle, which allows the algorithm to compute an arbitrary weighted combination of all rows of $A$ using a single query. A second difficulty is that we consider a very simple class of homogeneous bilinear functions, which means that any lower bound necessarily has to be of this form (as opposed to more general min-max optimization problems, where for lower bounds we can use functions with a richer structure). To handle these difficulties, we introduce a different proof technique, which still forces the algorithm to discover information about $A$ in an incremental manner, but in terms of certain vector outer products rather than rows. Specifically, we will utilize the following randomized construction of $A$:

\begin{lbconstruction}\label{assump:Aconst}
Given an iteration bound $T\geq 1$ such that $2T+1\leq \min\{n,d\}$,  positive parameters $\alpha,\beta$, and an algorithm $\Acal$ interacting with a two-sided oracle for $T$ iterations, the matrix $A$ is defined as
\[
A~=~\sum_{j=1}^{T}(\bv_{j-1}-\bv_j)\bu_j^\top+\bv_T\bu_{T+1}^{\top}~,
\]
where $\forall j$, $\bv_j\in \reals^n$, $\bu_j\in \reals^d$, and these vectors are constructed iteratively as follows:
\begin{itemize}
    \item Let $\bv_0=\alpha \mathbf{1}$ (where $\mathbf{1}$ is the all-ones vector).
    \item For $t=1,2,\ldots,T+1:$
    \begin{itemize}
        \item Compute algorithm $\Acal$'s queries $\bp_t,\bw_t$ (based on oracle outputs received so far).
\item Let $\bv_t=\beta\cdot \Pi_{t}(\bxi_t)$ be the scaled projection of an independent Gaussian $\bxi_t\sim\Ncal(\mathbf{0},I_n)$, where $\Pi_t:\reals^n\to\mathrm{span}(\bv_1,\ldots,\bv_{t-1},\bp_1,\ldots,\bp_t)^\perp$
is the projection onto the subspace orthogonal to
$\bv_1,\ldots,\bv_{t-1},\bp_1,\ldots,\bp_t$.

\item Let $\bu_t=\frac{1}{\|\Pi'_{t}(\bxi'_t)\|_2}\Pi'_{t}(\bxi'_t)$ be the normalized projection of an independent Gaussian $\bxi'_t\sim\Ncal(\mathbf{0},I_d)$ (also independent of $\bxi_t$)
where $\Pi'_t:\reals^d\to\mathrm{span}(\bu_1,\ldots,\bu_{t-1},\bw_1,\ldots,\bw_t)^\perp$
is the projection onto the subspace orthogonal to $\bu_1,\ldots,\bu_{t-1},\bw_1,\ldots,\bw_t$.
\item Let $A_t = \sum_{j=1}^{t}(\bv_{j-1}-\bv_j)\bu_j^\top$, and feed $\Acal$ with $\bp_t^\top A_t$ and $A_t\bw_t$ (as a response to its queries $\bp_t,\bw_t$).
    \end{itemize}
\end{itemize}
\end{lbconstruction}

Note that by construction, $\bv_t$ is
a standard Gaussian random vector in $\reals^n$, scaled by $\beta$ and projected on the subspace orthogonal to $\mathrm{span}(\bv_0,\ldots,\bv_{t-1},\bp_1,\ldots,\bp_t)$, hence, $\bv_t$ is orthogonal to all these vectors. This is a non-trivial projection (i.e., not zero)
due to the assumption $2T+1\leq n$. Similarly, the assumption $2T+1\leq d$ ensures that the construction of $\bu_1,\dots,\bu_{T+1}\in\reals^d$ results in an orthonormal set (with probability $1$). The choice of a Gaussian distribution is crucially used in order to control the sign and magnitudes of various quantities associated with the resulting matrix, as will be seen later in the proof.

For the construction to be valid, we need to ensure that the oracle responses as defined above, using intermediate matrices $A_t$, are all consistent with the same fixed matrix $A$ in hindsight. This is formalized in the following lemma, whose proof easily follows from the orthogonality properties in the construction, as detailed in the proof section.

\begin{lemma}[Responses simulate oracle on $A$]\label{lem:sim}
    For all $t\in [T]:~\bp_t^\top A_t=\bp_t^\top A$ and $A_t\bw_t=A\bw_t$. Therefore, the sequences of vectors $\{\bp_t\}_{t=1}^{T},\{\bw_t\}_{t=1}^{T+1}$ are the same as if the algorithm was fed with the oracle $\Ocal_2^{A}$ for $T$ iterations. Moreover,
    \begin{equation} \label{eq: Aw_{T+1}}
        A\bw_{T+1}~=~\sum_{j=1}^{T}(\bv_{j-1}-\bv_j)\bu_j^\top \bw_{T+1}~.
    \end{equation}
\end{lemma}

A second crucial requirement for the construction is that the resulting matrix $A$ defines a game with some (sufficiently large) positive value.
This is formalized in the following lemma:

\begin{lemma}\label{lem:sep}
For any $\delta\in(0,1)$ and $A$ as defined in Construction \ref{assump:Aconst}, 
there is an absolute constant $c>0$ such that with probability at least $1-\delta:$
    \[
    \max_{\bw:\norm{\bw}_p\leq 1}\min_{\bp\in \Delta^{n-1}}\bp^\top A\bw~\geq~
    c\,\alpha\cdot 
    \begin{cases}
     d^{\half-\frac{1}{p}}/\sqrt{T} & \text{if~~}p\in[1,2]
    \\
     d^{\half-\frac{1}{p}}/\sqrt{T\log(d/\delta)} & \text{if~~}p\in(2,\infty]
    \end{cases}
    ~.
    \]
\end{lemma}

The proof of \lemref{lem:sep} is based on noting that $A\bw=\frac{\alpha}{\|\sum_{t=1}^{T+1}\bu_t\|_p}\mathbf{1}$
for $\bw:=\frac{1}{\|\sum_{t=1}^{T+1}\bu_t\|_p}\sum_{t=1}^{T+1}\bu_t$, which is a vector  in the $\ell_p$ ball,
and hence the maximum over the ball is at least $\frac{\alpha}{\|\sum_{t=1}^{T+1}\bu_t\|_p}$. We lower bound this value with high probability, based on Gaussian properties of the construction.

A final consistency requirement for the construction is that we need to choose the $\alpha,\beta$ parameters appropriately, to satisfy the constraint that the operator norm $\|A\|_{p\to\infty}$ (equivalently, the largest $\ell_q$-norm over $A$'s rows where $1/p+1/q=1$) is at most $1$, at least with high probability, implying that a suitable choice of $A$ exists. This can be ensured via the following lemma:

\begin{lemma}\label{lem: opnorm}
For any $\delta\in(0,1)$ and $A$ as defined in Construction \ref{assump:Aconst}, there is an absolute constant $C'>0$ such that with probability at least $1-\delta-\exp(-d/48):$
\begin{align*}
    \|A\|_{p\to\infty}\leq C'\cdot
    d^{\half-\frac{1}{p}}\sqrt{\log(nd/\delta)}\left[\alpha+\beta\sqrt{T}+\beta\sqrt{\log(nd/\delta)}\right]
    ~.
\end{align*}
Thus, as long as $\alpha,\beta,\delta$ satisfy 
\[
d^{\half-\frac{1}{p}}\sqrt{\log(nd/\delta)}\left[\alpha+\beta\sqrt{T}+\beta\sqrt{\log(nd/\delta)}\right]\leq\frac{1}{C'}~,
\]
it holds with probability at least $1-\delta-\exp(-d/48)$ that $ \|A\|_{p\to\infty}\leq 1$. 
\end{lemma}

The proof of the lemma appears in the proof section, and follows from Gaussian concentration properties as well. 
With these consistency components in place, the main task now is to show that the algorithm's output $\bw_{T+1}$ cannot output a vector with positive value (at least with high probability over the randomized choice of $A$, which implies that a suitable $A$ exists). This is formalized in the following proposition:

\begin{proposition}[Algorithm returns a non-separator]\label{prop:nosep}
Suppose $A$ is constructed as in Construction \ref{assump:Aconst}, and that $4\alpha/\beta\leq 1/\sqrt{T}$. Then for any $\delta\in (0,1)$, if $T\sqrt{80\log(2T/\delta)/n}\leq \frac{1}{4}$, then with probability at least $1-\delta-\exp\left(T\log(2n)-n/32\right)$ over the choice of $\bxi_1,\ldots,\bxi_T$, it holds that
\[
\sup_{\bw\in \reals^d}~\min_{\bp\in \Delta^{n-1}}~\bp^\top\sum_{j=1}^{T}(\bv_{j-1}-\bv_j)\bu_j^\top \bw~\leq~ 0~,
\]
and thus by \eqref{eq: Aw_{T+1}},
\[
\min_{\bp\in \Delta^{n-1}}\bp^\top A \bw_{T+1}~\leq~0~.
\]
\end{proposition}

The formal proof of \propref{prop:nosep} is rather involved. At a high level, we use Gaussian concentration properties and an $\epsilon$-net argument, to show that with high probability over the randomized construction, it holds for any $\bw$ that the vector $\sum_{j=1}^{T}(\bv_{j-1}-\bv_j)\bu_j^\top \bw$ contains an entry which is upper bounded by a certain expression, which we then prove to be non-positive. Therefore, $\min_{\bp\in\Delta}\bp^\top\sum_{j=1}^{T}(\bv_{j-1}-\bv_j)\bu_j^\top \bw$ is non-positive for all $\bw$.

We can now assemble the claims outlined throughout this section,
and set the parameters $\alpha,\beta$ appropriately in order to prove \thmref{thm:main}:

\begin{proof}[Proof of \thmref{thm:main}]

 We construct the matrix $A$ as in Construction \ref{assump:Aconst}. 
Examining the conditions in Lemmas \ref{lem:sep}, \ref{lem: opnorm} and \propref{prop:nosep}, and ignoring log factors momentarily, we see that in order for all of them to be applicable,
we must ensure $\alpha\lesssim \beta/\sqrt{T}$ (\propref{prop:nosep}), satisfy 
$\beta\lesssim d^{1/p-1/2}/\sqrt{T}$
and
$\alpha\lesssim d^{1/p-1/2}$
(\lemref{lem: opnorm}), while the game value guarantee is $\approx \alpha d^{1/2-1/p}/\sqrt{T}$
(\lemref{lem:sep}). Thus, to make the game value as large
as possible, under the constraints we should choose $\beta\approx d^{1/p-1/2}/\sqrt{T}$ and $\alpha\approx \beta/\sqrt{T}\approx d^{1/p-1/2}/T$, thus leading to a value gap of roughly $\alpha d^{1/2-1/p}/\sqrt{T}\approx 1/T^{3/2}$.

More formally, pick $\delta=1/8$ in Lemmas \ref{lem:sep}, \ref{lem: opnorm} and \propref{prop:nosep}, and
for the universal constant $C'$ given by \lemref{lem: opnorm}, we choose $\beta=\frac{1}{4C'd^{1/2-1/p}(\sqrt{T\log(8nd)}+\log(8nd))}$, as well as $\alpha=\frac{\beta}{4\sqrt{T}}=\frac{1}{16C'd^{1/2-1/p}\sqrt{T}(\sqrt{T\log(8nd)}+\log(8nd))}$. It is easily verified that with this choice, as well as the theorem assumptions, the conditions in Lemmas \ref{lem:sep}, \ref{lem: opnorm} and \propref{prop:nosep} are all satisfied, and with a union bound, the resulting matrix $A$ satisfies simultaneously with some positive probability  $\|A\|_{p\to\infty}\leq 1$ and $\min_{\bp\in\Delta^{n-1}}\bp^\top A\bw_{T+1}\leq 0$, as well as
\begin{align*}
\max_{\bw:\|\bw\|_p\leq 1}\min_{\bp\in\Delta^{n-1}}\bp^\top A\bw &\geq 
 c\,\alpha\cdot 
    \begin{cases}
     d^{\half-\frac{1}{p}}/\sqrt{T} & \text{if~~}p\in[1,2]
    \\
     d^{\half-\frac{1}{p}}/\sqrt{T\log(d/\delta)} & \text{if~~}p\in(2,\infty]
    \end{cases}
    \\&=
\frac{c}{16C'T(\sqrt{T\log(8nd)}+\log(8nd))}
\cdot 
    \begin{cases}
    1 & \text{if~~}p\in[1,2]
    \\
     1/\sqrt{\log(8d)} & \text{if~~}p\in(2,\infty]
    \end{cases}
~.
\end{align*}
Hence, by the probabilistic method, a suitable fixed matrix $A$ satisfying all of the above must exist, and the result simplifies to the statement in Theorem~\ref{thm:main} for sufficiently large $C>0$.
\end{proof}

\begin{remark}[Number of oracle queries vs. number of matrix-vector multiplications]
    As defined, the two-sided oracle allows two matrix-vector multiplications (one from each side of the matrix $A$) in each oracle call. Thus, up to a factor of $2$, any complexity lower bound for a two-sided oracle automatically implies a lower bound on the required total number of matrix-vector multiplications (on either side of the matrix). However, we note that our proof technique can be potentially applied to get a slightly stronger result: Namely, a similar lower bound on the number of \emph{alternations} between a right matrix-vector multiplication and a left matrix-vector multiplication, assuming the matrix size is large enough.\footnote{Specifically, consider a model where the algorithm performs a single matrix-vector multiplication in each iteration.  Then we can modify Construction \ref{assump:Aconst}, so that if there is a sequence of right matrix-vector multiplications using $\bw_{t_1},\ldots,\bw_{t_2}$, we can simply pick a single vector $\bu_{t_1}$ which is orthogonal to all of them, rather than constructing a sequence of vectors $\bu_{t_1},\ldots,\bu_{t_2}$. Similarly, if there is a sequence of left matrix-vector multiplications using $\bp_{t_1},\ldots,\bp_{t_2}$, we can pick a single $\bv_{t_1}$ which is orthogonal to all of them, rather than constructing a sequence of vectors $\bv_{t_1},\ldots,\bv_{t_2}$. Thus, the total number of vectors $\bu_{1},\bv_1,\bu_{2},\bv_2,\ldots$, and hence ultimately the lower bound, does not scale with the total number of matrix-vector multiplications performed by the algorithm, but rather with the number of alternations between left and right matrix-vector multiplications.}
\end{remark}

\section{Applications}

In this section, we discuss some key implications of our results.

\subsection{Nash equilibria}

A prominent problem class of interest is that of simplex/simplex games, which are well-known to model zero-sum games. Given a matrix with bounded entries $A\in[-1,1]^{n\times d}$, or equivalently $\|A\|_{1\to\infty}\leq 1$, the problem is to solve
\begin{equation} \label{eq: simplex/simplex}
    \max_{\bw\in\Delta^{d-1}}\min_{\bp\in\Delta^{n-1}}\bp^\top A\bw~.
\end{equation}
Recall that for this problem, a two-sided oracle is exactly equivalent to a first-order (or gradient) oracle, which returns the gradient of $(\bp,\bw)\mapsto \bp^\top A\bw $. Thus, a lower bound under the two-sided oracle model is precisely a lower bound for solving zero-sum games using a first-order oracle.

To apply \thmref{thm:main}, which holds for $\ell_p$ balls, we provide a reduction from simplex/simplex games to $\ell_1$/simplex games, showing that any lower bound for the latter implies a lower bound for the former. In other words, any algorithm for solving Problem~(\ref{eq: simplex/simplex}) (over the simplex) can be readily converted to an algorithm for solving $\max_{\bw:\|\bw\|_1\leq 1}\min_{\bp\in\Delta^{n-1}}\bp^\top A\bw$ (in a slightly modified dimension).

\begin{proposition}\label{prop:l1tosimp}
Suppose there is an algorithm $\Acal$ that for any matrix $A\in [-1,+1]^{n\times 2d}$, after interacting with $\Ocal_2^{A}$ for $T$ iterations, returns a vector $\bw_{T+1}\in \Delta^{2d-1}$ such that
    \[
\left(\max_{\bw\in\Delta^{2d-1}}\min_{\bp\in\Delta^{n-1}}\bp^\top A\bw\right)-\left(\min_{\bp\in\Delta^{n-1}}\bp^\top A\bw_{T+1}\right)
~\leq~ \epsilon(T,n,d)~.
    \]
    Then there is an algorithm such that for any $A\in [-1,+1]^{n\times d}$, after interacting with $\Ocal_2^{A}$ for $T$ iterations, it returns a vector $\bw_{T+1}\in\reals^d,\norm{\bw_{T+1}}_1\leq 1$ such that
    \[
\left(\max_{\bw\in\reals^d:\norm{\bw}_1\leq 1}\min_{\bp\in\Delta^{n-1}}\bp^\top A\bw\right)-\left(\min_{\bp\in\Delta^{n-1}}\bp^\top A\bw_{T+1}\right)
~\leq~\epsilon(T,n,d)~.
    \]
\end{proposition}

The formal proof is deferred to the proof section. In a nutshell, the idea is that given a matrix $A\in [-1,+1]^{n\times d}$ and an algorithm for a simplex/simplex matrix game, we can feed the algorithm with the matrix $(A;-A)\in \reals^{n\times 2d}$, and convert the algorithm's output $\bw$ (a vector in the $(2d-1)$-simplex) to a $d$-dimensional vector $\bw'$ in the $\ell_1$ unit ball, so that $A\bw'=(A;-A)\bw$, leading to similar guarantees. We note that the reduction does require us to modify the matrix width $d$ by a factor of $2$, but this will not affect our lower bound by more than a small constant factor (as the bound we will show applies to any sufficiently large $d$). We also note that the computational complexity of the two algorithms in the reduction are essentially identical. 

Combining \thmref{thm:main} for $p=1$, with \propref{prop:l1tosimp}, results in the following corollary:

\begin{corollary}
No deterministic first-order algorithm can guarantee returning a $\epsilon$-approximate Nash equilibrium of a matrix $A\in[-1,1]^{n\times d}$ (for sufficiently large $d,n\in\mathrm{poly}(1/\epsilon)$),
using less than $\tilde{\Omega}(\epsilon^{-2/3})$ calls to a two-sided oracle $\Ocal^A_2$. Moreover, as this lower bound is nearly-matched by \citet{karmarkar2026solving}, our result settles the first-order oracle complexity of this task (up to log factors).
\end{corollary}

Our lower bound for approximating Nash equilibria is nearly-matched by \citet{karmarkar2026solving}, hence our result settles the oracle complexity of this task (up to log factors) in the natural model of matrix-vector multiplications.

\subsection{Linear separability}

As previously mentioned, $\ell_2$/simplex games model the task of finding a linear separator: Given a linearly separable classification dataset $(\bx_l,y_l)_{l=1}^{n}\subset\reals^{d}\times\{\pm1\}$ with bounded norm $\max_{l\in[n]}\|\bx_l\|_2\leq 1$, consider the data matrix $A\in\reals^{n\times d}$ with $A_l:=y_l \bx_l$. Then $\|A\|_{2\to \infty}\leq 1$, and the problem
\[
\max_{\bw\in\reals^d:\|\bw\|_2\leq 1}\min_{\bp\in\Delta^{n-1}}\bp^\top A\bw
\]
corresponds to finding a linear separator with maximal margin $\gamma_A$.

To the best of our knowledge, no oracle complexity lower bounds for this task that appear in the current literature apply to accelerated algorithms, which as we showed in Theorem~\ref{thm:one-sided}, provably require two-sided multiplications of the data matrix $A$.
Applying \thmref{thm:main} for $p=2$ results in the following corollary:

\begin{corollary}
No deterministic algorithm can guarantee returning a linear separator of a classification dataset as above (for sufficiently large $d,n\in\mathrm{poly}(1/\gamma_A)$),
using less than $\tilde{\Omega}(\gamma_A^{-2/3})$ calls to a two-sided oracle $\Ocal^A_2$ of the data matrix.
\end{corollary}

Our lower bound for finding a linear separator is nearly-matched by \citet{karmarkar2026solving}, hence our result settles the oracle complexity of this task (up to log factors) in the natural model of matrix-vector multiplications.

\subsection{Combinatorial applications}

Another problem class which received notable interest in the literature is that of $\ell_\infty$/simplex games, also referred to as ``box-simplex games'', of the form
\[
\max_{\bw\in[-1,1]^d}\min_{\bp\in\Delta^{n-1}}\bp^\top A\bw~.
\]
Solving such games appears as an important subroutine in a variety of applications, mostly of combinatorial nature, such as the computing a graph's max-flow \citep{sherman2017area}, bipartite matching \citep{assadi2022semi},
densest subgraph \citep{boob2019faster,nguyen2024multiplicative}, as well as other computational tasks such as optimal transport \citep{jambulapati2019direct}.

\citet{sherman2017area} showed that given a $\ell_\infty$/simplex game with matrix $A$ satisfying $\|A\|_{\infty\to\infty}\leq 1$, an $\epsilon$-suboptimal solution can be returned using $O(\log(d)\epsilon^{-1}\log(\epsilon^{-1}))$ two-sided queries. To the best of our knowledge, no oracle complexity lower bound is known for this task. Applying Theorem~\ref{thm:main} for $p=\infty$ results in the following corollary:

\begin{corollary}
    No deterministic algorithm can guarantee returning an $\epsilon$-suboptimal solution to a $\ell_\infty$/simplex game with respect to a $n\times d$ matrix $A$ (for sufficiently large $d,n\in\mathrm{poly}(1/\epsilon)$), using less than $\tilde{\Omega}(\epsilon^{-2/3})$ calls to a two-sided oracle $\Ocal^A_2$.
\end{corollary}

Notably, as opposed to the previously discussed applications for $p\in\{1,2\}$, this result leaves open a $\tilde{O}(\epsilon^{-1/3})$ gap between the best-known upper bound and the lower bound we have established, leaving open an interesting problem for future work.

\section{Proofs} \label{sec: proofs}

\subsection{Proof of \lemref{lem: one sided}}

We start by noting that the case $p=2$ simply follows from orthonormality of $A$'s rows, since in that case it holds deterministically that
\[
\frac{\|A_l\|_2^2}{\|\sum_{t=1}^T A_t\|_2}
=\frac{\|A_l\|_2^2}{\sqrt{\sum_{t=1}^{T+1}\|A_t\|_2^2}}=\frac{1}{\sqrt{T+1}}~.
\]
We move on to prove the general probabilistic case, and denote $\bg_t:=\Pi_t(\bxi_t)$ so that $A_t=\frac{\bg_t}{\|\bg_t\|_q}$. Note that projecting a Gaussian results in a Gaussian with smaller (or equal) variance, and that in this case $\bxi_t$ is projected onto a subspace of dimension at least $d-2t\geq d-2T\geq \half d$. We can apply a standard Gaussian concentration bound, concretely the Laurent–Massart bound, to ensure that with probability at least $\frac{3}{4T}:$
\begin{align*}
\|\bg_t\|_2^2 \geq \half d - 2\sqrt{d\log(4T)/2}
\geq \frac{1}{4}d~,
\end{align*}
an event which further implies that
\begin{align*}
\|\bg_t\|_q \geq d^{1/q-1/2}\|\bg_t\|_2 \geq \half d^{1/q}~.
\end{align*}
Furthermore, a standard Gaussian tail bound $\Pr(|g_{t,i}|> z)\leq 2\exp(-z^2/2)$ implies via a union bound that with probability at least $\frac{3}{4T}:$
\begin{align*}
\|\bg_t\|_q\leq d^{1/q} \sqrt{2\log(16Td)}
~~~\text{and}~~~
\|\bg_t\|^2_2\leq 2d \log(8Td)~.
\end{align*}
Union bounding over $t$, we get that with probability at least $\half:$
\begin{align*}
\|A_l\|_2^2 = \frac{\|\bg_l\|_2^2}{\|\bg_l\|_q^2}\geq \frac{\frac{1}{4}d}{2d^{2/q}\log(16Td)}
=\frac{d^{1-2/q}}{8\log(16Td)}~,
\end{align*}
and also by using the orthogonality of $A$'s rows,
\begin{align*}
\left\|\sum_{t=1}^{T+1}A_t\right\|_p
&\leq d^{1/p-1/2}\left\|\sum_{t=1}^{T+1}A_t\right\|_2
=d^{1/p-1/2}\sqrt{\sum_{t=1}^{T+1}\|A_t\|_2^2}
=d^{1/p-1/2}\sqrt{\sum_{t=1}^{T+1}\frac{\|\bg_t\|_2^2}{\|\bg_t\|_q^2}}
\\&\leq
d^{1/p-1/2} \sqrt{(T+1)\cdot \frac{2d\log(8Td)}{\frac{1}{4}d^{2/q}}}
\\&\leq
4d^{1/p-1/q}\sqrt{(T+1)\log(8Td)}~.
\end{align*}
Thus overall,
\begin{align*}
\frac{\|A_l\|_2^2}{\left\|\sum_{t=1}^{T+1}A_t\right\|_p}
&\geq \frac{d^{1-2/q}}{8\log(16Td)\cdot 4d^{1/p-1/q}\sqrt{(T+1)\log(8Td)}}
\\&=\frac{1}{32\sqrt{T+1}\cdot\log(16Td)\sqrt{\log(8Td)}}~,
\end{align*}
where we used the fact that $1-\frac{2}{q}-\frac{1}{p}+\frac{1}{q}=1-\frac{1}{q}-\frac{1}{p}=0$.
Crudely bounding $\log(8Td)\leq \log(16Td)\leq \log(16d^2)=2\log(4d)$ completes the proof.

\subsection{Proof of \lemref{lem:sim}}

By the orthogonality assumptions in the construction, for any $t$,
\[
\bp_t^\top A_t~=~\bp_t^\top\left(\sum_{j=1}^{t}(\bv_{j-1}-\bv_j)\bu_j^\top\right)~=~\bp_t^\top\left(\sum_{j=1}^{T}(\bv_{j-1}-\bv_j)\bu_j^\top+\bv_T\bu_{T+1}^{\top}\right)~=~\bp_t^\top A~. 
\]
Similarly,
\[
A_t\bw_t~=~\left(\sum_{j=1}^{t}(\bv_{j-1}-\bv_j)\bu_j^\top\right)\bw_t~=~\left(\sum_{j=1}^{T}(\bv_{j-1}-\bv_j)\bu_j^\top+\bv_T\bu_{T+1}^{\top}\right)\bw_t=A\bw_t~.
\]
Since $\bp_t,\bw_t$ at iteration $t$ are determined by the previous oracle calls, it follows by induction that the sequences of vectors $\{\bp_t\}_{t=1}^{T},\{\bw_t\}_{t=1}^{T+1}$ are those produced by the algorithm given access to the oracle $\Ocal_2^{A}$ for the matrix $A$. Finally, the expression for $A\bw_{T+1}$ follows from the definition of $A$, and the fact that $\bu_{T+1}$ is chosen to be orthogonal to $\bw_{T+1}$.

\subsection{Proof of \lemref{lem:sep}}
    Consider $\bw=\frac{1}{\norm{\sum_{t=1}^{T+1}\bu_t}_p}\sum_{t=1}^{T+1}\bu_t$, which by construction satisfies $\norm{\bw}_p\leq 1$. Since $\bu_1,\ldots,\bu_{T+1}$ are orthonormal, we have
    \begin{align*}
    A\bw &= \frac{1}{\left\|\sum_{t=1}^{T+1}\bu_t\right\|_p}\cdot\left(\sum_{j=1}^{T}(\bv_{j-1}-\bv_j)\bu_j^\top+\bv_T\bu_{T+1}^{\top}\right)\left(\sum_{t=1}^{T+1}\bu_t\right)\notag\\
    &=\frac{1}{\left\|\sum_{t=1}^{T+1}\bu_t\right\|_p}\cdot\left(\sum_{j=1}^{T}(\bv_{j-1}-\bv_j)+\bv_T\right)
    \\&=\frac{1}{\left\|\sum_{t=1}^{T+1}\bu_t\right\|_p}\cdot\bv_0
    \\&=\frac{\alpha}{\left\|\sum_{t=1}^{T+1}\bu_t\right\|_p}\cdot \mathbf{1}~.
    \end{align*}
Thus, each entry of $A\bw$ equals $\alpha/\|\sum_{t=1}^{T+1}\bu_t\|_p$, and since this holds for the $\bw$ we have chosen, the maximum of $\min_{\bp\in \Delta^{n-1}}\bp^\top A\bw$ over all vectors $\bw$ in the unit $\ell_p$ ball can only be larger.
Therefore, it remains to upper bound the quantity $\|\sum_{t=1}^{T+1}\bu_t\|_p$.

For $p\in[1,2]$, it suffices to recall that $\bu_1,\dots,\bu_{T+1}$ are orthonormal, and by comparing the $\ell_p$ and $\ell_2$ norms we see that
\[
\left\|\sum_{t=1}^{T+1}\bu_t\right\|_p
\leq
d^{\frac{1}{p}-\half}\left\|\sum_{t=1}^{T+1}\bu_t\right\|_2
=d^{\frac{1}{p}-\half}\sqrt{T+1}
\leq 2d^{\frac{1}{p}}\sqrt{T}
~,
\]
proving the first lower bound in the lemma (deterministically).

For $p\in(2,\infty)$ we use a probabilistic argument. Note that for any index $i\in[d]$, the sequence $u_{1,i},\dots,u_{T+1,i}$ is constructed so that for any $t\in[T+1]:~u_{t,i}\mid u_{1,i},\dots,u_{t-1,i}$ is a mean-zero sub-Gaussian random variable with sub-Gaussian norm bounded by $O(1/d)$ (note that the projected dimension is still at least on the order of $d$ by assumption that $T\ll d$). Hence, a martingale bound (cf. \citealp{shamir2011variant}) ensures that with probability at least $1-\delta/d$ it holds that $|\sum_{t\in[T+1]}u_{t,i}|\leq C\sqrt{T\log(d/\delta)/d}$
for some universal constant $C>0$. Union bounding over $i\in[d]$ so that the bound holds for all $i$ with probability at least $1-\delta$, under this probable event we get
\begin{align*}
\left\|\sum_{t=1}^{T+1}\bu_t\right\|_p
&=\left(\sum_{i\in[d]}\left|\sum_{t\in[T+1]}u_{t,i}\right|^p\right)^{1/p}
\leq \left(\sum_{i\in[d]}C^p\left(\frac{T\log(d/\delta)}{d}\right)^{p/2}\right)^{1/p}
\\&=C d^{1/p-1/2}\sqrt{T\log(d/\delta)}~,
\end{align*}
proving the claimed bound for $p\in(2,\infty)$.

For $p=\infty$, under the same probable event as before, it holds that $\|\sum_{t=1}^{T+1}\bu_t\|_\infty
=\max_{i\in[d]}|\sum_{t\in[T+1]}u_{t,i}|\leq C\sqrt{T\log(d/\delta)/d}$ as well.

\subsection{Proof of \lemref{lem: opnorm}}

Recall that $\|A\|_{p\to\infty}=\max_{l\in[n]}\|A_l\|_q$ where $q\in[1,\infty]$ is the dual exponent satisfying $1/p+1/q=1$, and so we set out to bound the $q$-norm of $A$'s rows.
For any $l\in[n]$, it holds that
\begin{align}
\|A_l\|_q&=\left\|\sum_{j=1}^{T}(v_{j-1,l}-v_{j,l})\bu_j^\top+v_{T,l}\bu_{T+1}^{\top}\right\|_q \nonumber
\\&\leq \left\|\sum_{j=1}^{T}(v_{j-1,l}-v_{j,l})\bu_j^\top\right\|_q
+\|v_{T,l}\bu_{T+1}^{\top}\|_q~.
\label{eq: Al 2 terms}
\end{align}
We turn to bound the terms above. For $l\in[n],i\in[d]$, we denote $D^j_l:=v_{j-1,l}-v_{j,l}$ and $S_{l,i}:=\sum_{j=1}^{T}D^j_{l}u_{j,i}$. Recall (as in the proof of \lemref{lem:sep}) that each $u_{j,i}$ is a sub-Gaussian, mean-zero random variable with sub-Gaussian norm bounded by $O(1/d)$, and therefore conditioned on $\bv_0,\dots,\bv_T$,
the random variable
$S_{l,i}$ has sub-Gaussian norm at most $O(\frac{1}{d})\cdot\sum_{j=1}^{T}(D_l^j)^2$
(this follows from a basic summation property for sub-Gaussian norms,
cf. \citealp{vershynin2018high}). Therefore, there exists a universal constant $C_1>0$ such that for any $z>0:$
\begin{align*}
\Pr\left(|S_{l,i}|\geq z\mid \bv_0,\dots,\bv_T\right)\leq 2\exp\left(-\frac{C_1dz^2}{\sum_{j=1}^{T}(D_l^j)^2}\right)~.
\end{align*}
The numerical inequality $(a-b)^2\leq 2(a^2+b^2)$ implies that $\sum_{j=1}^{T}(D^j_l)^2\leq 4\sum_{j=0}^{T}v_{j,l}^2=4\alpha^2+4\sum_{j=1}^{T}v_{j,l}^2$, and further note that while $v_{1,l},\dots,v_{T,l}$ are not independent, it does hold that conditionally $v_{t,l}\mid v_{0,l},\dots v_{t-1,l}\sim \Ncal(0,\sigma_{t,l}^2)$ for some $\sigma_{t,l}\leq \beta$ since projecting a Gaussian cannot increase its variance. Therefore, \lemref{lem: exp tail bound}
ensures that with probability at least $1-\delta:~\sum_{j=1}^{T}v_{j,l}^2\leq 8\beta^2(T+\log(1/\delta))$, so via a union bound, we get that with probability at least $1-2\delta$ it holds for all $z>0$ that
\[
\Pr\left(|S_{l,i}|\geq z\right)\leq 2\exp\left(-\frac{C_1dz^2}{4\alpha^2+32\beta^2(T+\log(1/\delta))}\right)~.
\]
Equivalently, choosing $z$ by equating the above to $\delta$, slightly simplifying the expression and union bounding over $l\in[n]$ and $i\in[d]$, we see that
for some sufficiently large absolute constant $C>0$ it holds with probability at least $1-\delta$ that
\begin{align} \label{eq: Sli}
|S_{l,i}|\leq C\sqrt{\frac{\log(nd/\delta)}{d}}\left[\alpha+\beta(\sqrt{T}+\sqrt{\log(nd/\delta)})\right]~,
\text{~~for all }l\in[n],\,i\in[d]~. 
\end{align}
Under this event, for any $q<\infty$, it holds that
\begin{align}
\left\|\sum_{j=1}^{T}(v_{j-1,l}-v_{j,l})\bu_j^\top\right\|_q
&=\left(\sum_{i=1}^{d}\left|S_{l,i}\right|^q\right)^{1/q}
\label{eq: Al 1st term}\\&\leq \left(\sum_{i=1}^{d}\left(C\sqrt{\frac{\log(nd/\delta)}{d}}\left[\alpha+\beta(\sqrt{T}+\sqrt{\log(nd/\delta)})\right]\right)^{q}\right)^{1/q}
\nonumber\\&= Cd^{1/q-1/2}\sqrt{\log(nd/\delta)}\left[\alpha+\beta(\sqrt{T}+\sqrt{\log(nd/\delta)})\right]~, \nonumber
\end{align}
thus bounding the first term in \eqref{eq: Al 2 terms}. 
We similarly derive the bound for $q=\infty$
by simply applying \eqref{eq: Sli}, since $\|\sum_{j=1}^{T}(v_{j-1,l}-v_{j,l})\bu_j^\top\|_\infty
=\max_{i\in[d]}|S_{l,i}|$.

To bound the second term in \eqref{eq: Al 2 terms}, we apply \lemref{lem:rescaledGaussian} twice to get that with probability at least $1-2\delta-\exp(-d/48):$
\begin{align*}
\norm{\bv_{T}}_{\infty}~\leq~\beta\sqrt{2\log(2n/\delta)}
\text{~~~~~and~~~~~}
\norm{\bu_{T+1}}_{\infty}\leq \sqrt{8\log(2d/\delta)/d}~,
\end{align*}
and under this event, for any $l\in[n]$ it holds that
\begin{align}
\|v_{T,l}\bu_{T+1}^{\top}\|_q
&=|v_{T,l}|\cdot \|\bu_{T+1}^{\top}\|_q
\leq \|\bv_{T}\|_{\infty}\cdot  d^{1/q}\|\bu_{T+1}^{\top}\|_\infty
\nonumber\\&\leq
4\beta d^{1/q-1/2}\sqrt{\log(2n/\delta)\log(2d/\delta)}~.
\label{eq: Al 2nd term}
\end{align}
Overall, plugging \eqref{eq: Al 1st term} and \eqref{eq: Al 2nd term} into \eqref{eq: Al 2 terms}
and noting that the second summand is dominated by the first (for sufficiently large constant $C>0$), rescaling $\delta$ by a factor of $3$ and slightly simplifying completes the proof.

\subsection{Proof of \propref{prop:nosep}}\label{subsec:proofnosep}

We start by introducing notation that will be used throughout the proof. Let $M_t\in\reals^{n\times \mathrm{dim}(\mathrm{span}(\bv_1,\dots,\bv_{t-1},\bp_1,\dots,\bp_t))}$ be a matrix whose columns are an orthonormal basis for $\mathrm{span}(\bv_1,\dots,\bv_{t-1},\bp_1,\dots,\bp_t)$, and note that $\bv_t=\beta(I-M_t M_t^T)\bxi_t$.

Fix some $\bw\in\reals^d$, and define the auxiliary vector $\br=(r_1,\ldots,r_T)\in\reals^{T}$ as
\begin{equation}\label{eq:rdef}
\forall j\in [T-1]:~r_j=\bu_{j}^\top\bw-\bu_{j+1}^\top\bw~~,~~ r_T = \bu_T^\top \bw~,
\end{equation}
as well as the scalar $r_0=\bu_1^\top \bw$.
By construction, we have
\begin{align}
\sum_{j=1}^{T}(\bv_{j-1}-\bv_j)\bu_j^\top \bw~&=~
\bv_0(\bu_1^\top\bw)-\sum_{j=1}^{T-1}\bv_j(\bu_{j}^\top\bw-\bu_{j+1}^\top\bw)-\bv_T(\bu_T^\top\bw)\notag\\
&=~r_0\bv_0-\sum_{j=1}^{T-1}r_j\bv_j-r_T\bv_T~=~ r_0\bv_0-\sum_{j=1}^{T}r_j\bv_j\notag\\
&=~\alpha r_0\mathbf{1}-\beta\sum_{j=1}^{T}r_j(I-M_j M_j^\top)\bxi_j\notag\\
&=~\alpha r_0\mathbf{1}-\beta\left(\sum_{j=1}^{T}r_j\bxi_j-\sum_{j=1}^{T}r_jM_jM_j^\top\bxi_j\right)~.
\label{eq:valvec}
\end{align}
Our goal will be to show that with high probability over the random choice of the Gaussian random variables $\bxi_1,\ldots,\bxi_T$, \emph{simultaneously} for any $\br$, the vector above contains a non-positive entry. Thus, with high probability, the expression $\sum_{j=1}^{T}(\bv_{j-1}-\bv_j)\bu_j^\top \bw$ will have a non-positive entry simultaneously for any $\bw$, which implies the proposition. For that, we will analyze separately $\sum_{j=1}^{T}r_j\bxi_j$ and $\sum_{j=1}^{T}r_jM_jM_j^\top\bxi_j$, in the following two lemmas.

\begin{lemma}\label{lem:purexi}
    The following holds with probability at least 
    \[
    1-\exp\left(T\log(2n)-\frac{n}{32}\right)
    \]
    over the random choice of $\bxi_1,\ldots,\bxi_T$: For any $\br\in\reals^T$, there exists a subset $\Lcal\subseteq[n]$ such that
    \[
    |\Lcal|\geq \frac{n}{20}~~~\text{and}~~~ \min_{l\in \Lcal}\left(\sum_{j=1}^{T}r_j\bxi_j\right)_{l}\geq \frac{1}{2}\norm{\br}~.
    \]
\end{lemma}
\begin{proof} 
    Let $\Xi\in\reals^{n\times T}$ be the matrix whose $j$-th column is $\bxi_j$, so that $\sum_{j=1}^{T}r_j\bxi_j=\Xi\br$. Note that $\Xi$ is composed of $n\times T$ independent standard Gaussian entries. 

    Without loss of generality, it is enough to prove the bound for any unit $\br$: Namely that with high probability, for any unit $\br$,
    \begin{equation}\label{eq:runit}
    \left(\Xi \br\right)_{l}~\geq~\frac{1}{2}~~~\text{for at least $\frac{n}{20}$ indices $l$}
    \end{equation}
     (the result for all $\br\in\reals^T$ follows immediately by scaling $\br$). First, let us fix a unit $\br$ and some index $l\in [n]$, and note that $(\Xi \br)_{l}$ has a standard univariate Gaussian distribution. A standard fact about this distribution is that the probability of getting more than $1$ (namely, more than one standard deviation from the mean) is at least $0.1587...\geq 0.15$. Therefore, for any fixed $\br,l$, 
    \[
    \E\left[\mathbf{1}_{(\Xi\br)_{l}\geq 1}\right]~=~\Pr((\Xi \br)_{l}\geq 1)> 0.15~,
    \]
    where $\mathbf{1}_E$ is the indicator function for the event $E$. 
    Noting that $\left\{\mathbf{1}_{(\Xi\br)_{l}\geq 1}\right\}_{l=1}^{n}$ are independent random variables (since the rows of $\Xi$ are independent), it follows by a standard multiplicative Chernoff bound that
\begin{align*}
    \Pr\left(\sum_{l=1}^{n}\mathbf{1}_{(\Xi\br)_{l}\geq 1}\leq \frac{n}{20}\right)
&=\Pr\left(\sum_{l=1}^{n}\mathbf{1}_{(\Xi\br)_{l}\geq 1}\leq 0.15n\cdot \frac{1}{3}\right)~\\
&\leq\exp\left(-\frac{2}{9}\cdot 0.15n\right)=\exp\left(-\frac{n}{30}\right)~.
\end{align*}
    This bound is true for any \emph{fixed} unit $\br$. In order to extend the bound simultaneously for all unit $\br$, we will use a standard $\epsilon$-net argument: Fix some $\epsilon>0$, and let $\Ncal_{\epsilon}$ be an $\epsilon$-net of the unit Euclidean ball in $\reals^T$, of size $\leq \left(\frac{2}{\epsilon}+1\right)^T$ (namely, a collection of unit vectors so that each $\br$ in the unit ball is $\epsilon$-close in Euclidean distance to one of the vectors). The existence of an $\epsilon$-net of this size is shown, for example, in \cite[Lemma 5.2]{vershynin2012introduction}. Using a union bound and the displayed equation above, it follows that with probability at least $1-\left(\frac{2}{\epsilon}+1\right)^T\exp\left(-\frac{n}{30}\right)$,
    \begin{equation}\label{eq:netgood}
    \forall \br\in\Ncal_{\epsilon}~,~\sum_{l=1}^{n}\mathbf{1}_{(\Xi\br)_{l}\geq 1}> \frac{n}{20}~.
    \end{equation}
    Assuming this event holds, let $\bs$ be any other unit vector in $\reals^T$, such that $\norm{\br-\bs}\leq \epsilon$ for some $\br\in\Ncal_{\epsilon}$. Then letting $\Xi_{l}$ be the $l$-th row of $\Xi$, we have
    \begin{equation}\label{eq:net}
    \max_{l\in[n]}\left|\left(\Xi\br-\Xi\bs\right)_{l}\right|~=~
    \max_{l\in [n]}\left|\Xi_{l}(\br-\bs)\right|~\leq~
    \max_{l\in[n]}\norm{\Xi_{l}}\cdot \norm{\br-\bs}
    ~\leq~ \epsilon \cdot \max_{l\in[m]}\norm{\Xi_{l}}~.
    \end{equation}
    Since $\Xi_{l}$ is a standard Gaussian random vector in $\reals^T$, it follows by a standard tail bound for the norm of such vectors that
    \[
    \forall l\in [n] \Pr\left(\norm{\Xi_{l}}> \frac{1}{2\epsilon}\right)~\leq~ 2\exp\left(-\frac{(1/2\epsilon)^2}{2T}\right)~=~2\exp\left(-\frac{1}{8\epsilon^2 T}\right)~,
    \]
    hence by a union bound,
    \[
    \Pr\left(\max_{l\in [n]}\norm{\Xi_{l}}> \frac{1}{2\epsilon}\right)~\leq~ 2n\exp\left(-\frac{1}{8\epsilon^2 T}\right)~.
    \]
    Combining this with \eqref{eq:net}, it follows that with probability at least $1-2n\exp\left(-\frac{1}{8\epsilon^2 T}\right)$, it holds simultaneously for any $\epsilon$-close unit vectors $\br,\bs$ that
    \[
    \max_{l\in[n]}\left|\left(\Xi\br-\Xi\bs\right)_{l}\right|\leq\frac{1}{2}~.
    \]
    Combining the above with \eqref{eq:netgood} using a union bound, it follows that with probability at least 
    \[
    1-\left(\frac{2}{\epsilon}+1\right)^T\exp\left(-\frac{n}{30}\right)-2n\exp\left(-\frac{1}{8\epsilon^2 T}\right)~,
    \]
    it holds simultaneously for all unit vectors $\bs$ in $\reals^T$ that
    \[
    \sum_{l=1}^{n}\mathbf{1}_{(\Xi\bs)_{l}\geq \frac{1}{2}}> \frac{n}{20}~.
    \]
    In particular, choosing $\epsilon=\frac{2}{n}$, we get that the probability of this event not occuring is at most
    \[
        (n+1)^T\exp\left(-\frac{n}{30}\right)+2n\exp\left(-\frac{n^2}{32T}\right)~.
    \]
    Assuming $n\geq T$ (without loss of generality, since otherwise the probability lower bound in the lemma statement is less than $0$),  this can be loosely upper bounded by
    \[
    (n+1)^{T}\exp\left(-\frac{n}{32}\right)+2n\exp\left(-\frac{n}{32}\right)~=~\left((n+1)^{T}+2n\right)\exp\left(-\frac{n}{32}\right)
    ~\leq~ 2(2n)^T\exp\left(-\frac{n}{32}\right)~,
    \]
    which equals $2\exp(T\log(2n)-\frac{n}{32})$. This establishes \eqref{eq:runit} for any unit $\br$ as required. 
\end{proof}

\begin{lemma}\label{lem:remainder}
    For any $\delta\in (0,1)$, it holds with probability at least $1-\delta$ over $\bxi_1,\ldots,\bxi_T$ that 
    \[
    \forall \br\in \reals^T,~~\left\|\sum_{j=1}^{T}r_j M_j M_j^\top \bxi_j\right\|^2~\leq~ 4\norm{\br}^2T^2\log\left(\frac{2T}{\delta}\right)~.
    \]
\end{lemma}
\begin{proof}
    It is enough to prove that with probability at least $1-\delta$,
    \begin{equation}\label{eq:bnd1}
    \forall \br\in \reals^d:\norm{\br}=1,~~\left\|\sum_{j=1}^{T}r_j M_j M_j^\top \bxi_j\right\|^2~\leq~ 4T^2\log\left(\frac{2T}{\delta}\right)~
    \end{equation}
    (the result for all $\br$ then follow by scaling). 

    By the triangle inequality and Cauchy-Schwartz, we have for any unit vector $\br$ 
    \begin{align}
    \left\|\sum_{j=1}^{T}r_j M_j M_j^\top \bxi_j\right\|^2
    ~&\leq~ \left(\sum_{j=1}^{T}|r_j|\norm{M_j M_j^\top \bxi_j}\right)^2~\leq~
    \left(\sum_{j=1}^{T}r_j^2\right)\cdot \left(\sum_{j=1}^{T}\norm{M_j M_j^\top \bxi_j}^2\right)\notag\\
    &=~ \sum_{j=1}^{T}\norm{M_j M_j^\top \bxi_j}^2~.\label{eq:bnd2}
    \end{align}
    Since $\bxi_j$ is a standard Gaussian random vector (with zero mean and identity covariance), each $M_j M_j^\top \bxi_j$ (conditioned on $\bxi_1,\ldots,\bxi_{j-1}$ which in turn fixes $M_j$) is also Gaussian, with zero mean and covariance $(M_j M_j^\top)^2=M_j M_j^\top$. Therefore, by a standard tail bound for Gaussian random vectors,
    \begin{align*}
    \Pr\left(\norm{M_j M_j^\top \bxi_j}\geq z~|~\bxi_1,\ldots,\bxi_{j-1}\right)
    &\leq 2\exp\left(-\frac{z^2}{2 \text{Tr}(M_j M_j^\top)}\right)
    =2\exp\left(-\frac{z^2}{2\|M_j\|_F^2}\right)
    \\&\leq 2\exp\left(-\frac{z^2}{4j}\right)~,
    \end{align*}
    where $\|\cdot\|_F$ denotes Frobenius norm, and where the last step follows from the fact that $M_j$ is a matrix composed of at most $2j$ orthonormal rows. Therefore, for any $\delta'\in (0,1)$, it holds with probability at least $1-\delta'$ over $\bxi_t$ that
    \[
    \norm{M_j M_j^\top \bxi_j}\leq \sqrt{4j\log(2/\delta')}~.
    \]
    Letting $\delta'=\delta/T$ and using a union bound over all $j\in [T]$, we get that with probability at least $1-\delta$, \eqref{eq:bnd2} is at most
    \[
    \sum_{j=1}^{T}4j\log(2T/\delta)~\leq ~\sum_{j=1}^{T}4T\log(2T/\delta)
    ~=~ 4T^2\log(2T/\delta)~,
    \]
    which leads to \eqref{eq:bnd1} as required.
    \end{proof}

Combining \lemref{lem:purexi} and \lemref{lem:remainder} with a union bound, and applying them to \eqref{eq:valvec}, we get the following: With probability at least $1-\delta-\exp\left(T\log(2n)-\frac{n}{32}\right)$ over $\bxi_1,\ldots,\bxi_T$, we have that for any $\bw$, there is some subset $\Lcal\subseteq [n]$ of size $|\Lcal|\geq \frac{n}{20}$ for which the following inequalities hold:
\begin{align}
\min_{l\in[n]}\left(\sum_{j=1}^{T}(\bv_{j-1}-\bv_j)\bu_j^\top \bw\right)_{l}~&=~
\min_{l\in[n]}
\left(\alpha r_0\mathbf{1}-\beta\left(\sum_{j=1}^{T}r_j\bxi_j-\sum_{j=1}^{T}r_jM_jM_j^\top\bxi_j\right)\right)_{l}\notag\\
&\leq~
\min_{l\in\Lcal}
\left(\alpha r_0-\beta\left(\sum_{j=1}^{T}r_j\bxi_j\right)_{l}+\beta\left(\sum_{j=1}^{T}r_jM_jM_j^\top\bxi_j\right)_{l}\right)\notag\\
&\stackrel{(1)}{\leq}~
\min_{l\in\Lcal}
\left(\alpha r_0-\frac{\beta}{2}\norm{\br}+\beta\left(\sum_{j=1}^{T}r_jM_jM_j^\top\bxi_j\right)_{l}\right)\notag\\
&\stackrel{(2)}{\leq}~
\alpha r_0-\frac{\beta}{2}\norm{\br}+\beta \sqrt{\frac{1}{|\Lcal|}\sum_{l\in \Lcal}\left(\sum_{j=1}^{T}r_jM_jM_j^\top\bxi_j\right)_{l}^2}\notag\\
&\leq~
\alpha r_0-\frac{\beta}{2}\norm{\br}+\beta \sqrt{\frac{1}{|\Lcal|}\sum_{l=1}^{n}\left(\sum_{j=1}^{T}r_jM_jM_j^\top\bxi_j\right)_{l}^2}\notag\\
&\stackrel{(3)}{\leq}~\alpha r_0-\frac{\beta}{2}\norm{\br}+\beta \sqrt{\frac{20}{n}\cdot4\norm{\br}^2 T^2\log\left(\frac{2T}{\delta}\right)}\notag\\
&=~\alpha r_0-\beta\left(\frac{1}{2}-T\sqrt{\frac{80\log(2T/\delta)}{n}}\right)\norm{\br}\notag\\
&\stackrel{(4)}{\leq}~\alpha r_0-\frac{\beta}{4}\norm{\br}\notag\\
&\stackrel{(5)}{=}~ \frac{\beta}{4}\left(\frac{4\alpha}{\beta}\bu_1^\top\bw-\sqrt{\sum_{j=1}^{T-1}(\bu_{j}^\top\bw-\bu_{j+1}^\top\bw)^2+(\bu_T^\top \bw)^2}\right)~,\label{eq:finalexp}
\end{align}
where $(1)$ is by \lemref{lem:purexi}, $(2)$ is by the fact that a minimum can be upper bounded by an average, $(3)$ is by \lemref{lem:remainder} and the fact that $|\Lcal|\geq \frac{n}{20}$, $(4)$ is by the assumption in the proposition statement, and $(5)$ is by definition of the vector $\br$ and scalar $r_0$ in \eqref{eq:rdef}.

We now wish to argue that the expression in \eqref{eq:finalexp} is necessarily non-positive, which would imply overall that 
\[
\min_{\bp\in\Delta^{n-1}}\left(\sum_{j=1}^{T}(\bv_{j-1}-\bv_j)\bu_j^\top \bw\right)~=~\min_{l\in[n]}\left(\sum_{j=1}^{T}(\bv_{j-1}-\bv_j)\bu_j^\top \bw\right)_{l}~\leq~0~.
\]
Since $\bw$ is arbitrary, and the probabilistic statements above hold with high probability simultaneously for any vectors $\bw,\br$, it follows that with this high probability, 
\[
\sup_{\bw\in\reals^d}\min_{\bp\in\Delta^{n-1}}\left(\sum_{j=1}^{T}(\bv_{j-1}-\bv_j)\bu_j^\top \bw\right)~=~\min_{l\in[n]}\left(\sum_{j=1}^{T}(\bv_{j-1}-\bv_j)\bu_j^\top \bw\right)_{l}~\leq~0~,
\]
hence proving the proposition. 

Indeed, the fact that \eqref{eq:finalexp} is non-positive follows from the proposition assumption that $\frac{4\alpha}{\beta}\leq \frac{1}{\sqrt{T}}$, and the following lemma:
\begin{lemma}
    For any integer $T>1$ and $\delta\in \left(0,\frac{1}{\sqrt{T}}\right]$, it holds that
    \[
    \sup_{\bx\in\reals^T} \delta x_1-\sqrt{\sum_{j=1}^{T-1}(x_{j+1}-x_j)^2+x_T^2}~\leq~0~.
    \]
\end{lemma}
\begin{proof}
    Let $M\in\reals^{T\times T}$ be the symmetric matrix defined as
    \[
    M~:=~\begin{pmatrix}
1 & -1 & 0 & 0 & 0 &\cdots & 0\\
-1 & 2 & -1 & 0 & 0 &\cdots & 0\\
0 & -1 & 2 & -1 & 0 & \cdots & 0\\
\vdots &\vdots &\vdots &\vdots &\vdots &\vdots &\vdots\\
0 &\cdots &\cdots &0 &-1&2&-1\\
0 &\cdots &\cdots &0 &0&-1&2\\
\end{pmatrix}~.
    \]
    It is easily verified that the expression in the lemma statement equals $\delta \be_1^\top \bx-\sqrt{\bx^\top M \bx}$, where $\be_1$ is the first standard basis vector. Moreover, $M$ is positive semidefinite, as by definition $\bx^\top M\bx = \sum_{j=1}^{T-1}(x_{j+1}-x_j)^2+x_T^2\geq 0$ for all $\bx$. 
    
    Suppose by contradiction that the lemma does not hold, namely there exists some $\bx\in\reals^T$ such that 
    \begin{equation}\label{eq:suppos}
        \delta \be_1^\top \bx-\sqrt{\bx^\top M \bx}~>~0~.     
    \end{equation}
    Since $\sqrt{\bx^\top M\bx}\geq 0$, it follows that $\delta\be_1^\top\bx>0$, and therefore $\delta \be_1^\top \bx+\sqrt{\bx^\top M \bx}> 0$. As a result,
    \begin{align*}
        0~&<~ \left(\delta \be_1^\top \bx-\sqrt{\bx^\top M \bx}\right)
        \left(\delta \be_1^\top \bx+\sqrt{\bx^\top M \bx}\right)\\
        &=~ \delta^2(\be_1^\top \bx)^2-\bx^\top M \bx
        ~=~ \bx^\top\left(\delta^2 \be_1\be_1^\top-M\right)\bx~.
    \end{align*}
    In other words, we get that $\bx^\top (M-\delta^2\be_1\be_1^\top)\bx<0$, and therefore $M-\delta^2\be_1\be_1^\top$ is \emph{not} a positive semidefinite matrix. However, we will now show that $M-\delta^2\be_1\be_1^\top$ is positive semidefinite, which leads to a contradiction, hence \eqref{eq:suppos} cannot hold, and thus proving the lemma. Indeed, for any vector $\by$, we have
    \begin{align*}
        \by^\top (M-\delta^2\be_1\be_1^\top)\by&=~-\delta^2 y_1^2+\sum_{j=1}^{T-1}(y_j-y_{j+1})^2+y_T^2\\
        &\stackrel{(1)}{\geq}~ -\delta^2 y_1^2+\frac{1}{T}\left(\sum_{j=1}^{T-1}|y_j-y_{j+1}|+|y_T|\right)^2\\
        &\stackrel{(2)}{\geq}~-\delta^2 y_1^2+\frac{1}{T}\left(\sum_{j=1}^{T-1}(y_j-y_{j+1})+y_T\right)^2\\
        &=~-\delta^2 y_1^2+\frac{1}{T}y_1^2~=~
        \left(-\delta^2+\frac{1}{T}\right)y_1^2\\
        &\stackrel{(3)}{\geq}~0~,
    \end{align*}
    where $(1)$ uses the fact that $\norm{\bz}_2\geq \frac{1}{\sqrt{T}}\norm{\bz}_1$ for any $\bz\in \reals^{T}$, $(2)$ uses the triangle inequality, and $(3)$ is by the assumption that $\delta\in \left(0,\frac{1}{\sqrt{T-1}}\right)$. Therefore, $M-\delta^2\be_1\be_1^\top$ is a positive semidefinite matrix, which as explained above proves the lemma.
\end{proof}

\subsection{Proof of Proposition \ref{prop:l1tosimp}}
    Given any matrix $A\in [-1,+1]^{n\times d}$, define the operator 
    \[
    \psi(A):=(A;-A)\in \reals^{n\times 2d}~.
    \]
    Somewhat abusing notation, we will also define $\psi$ as a mapping from the unit $\ell_1$ ball in $\reals^d$ to $\Delta^{2d-1}$ as follows:
    \[
    \psi(\bw) = \sum_{j=1}^{d}|w_j|\left(\mathbf{1}_{w_j>0}\cdot\be_j + \mathbf{1}_{w_j\leq 0}\cdot\be_{j+d}\right)~\in~\reals^{2d}~,
    \]
    where $\mathbf{1}_E$ is the indicator function for the event $E$.
    In words, $\psi(\bw)$ is the $2d$-dimensional vector where the first $d$ entries correspond to the positive entries in $\bw$, and the last $d$ entries correspond to the absolute values of the non-positive ones. We will also define the inverse operator $\psi^{-1}$ from $\Delta^{2d-1}$ to the unit $\ell_1$ ball in $\reals^d$ as
    \[
    \psi^{-1}(\bw):=\sum_{j=1}^{d}(w_{j}-w_{d+j})\be_j\in\reals^d~.
    \]
    This is indeed a mapping to the unit $\ell_1$ ball, since for any $\bw\in\Delta^{2d-1}$,
    \[\norm{\psi^{-1}(\bw)}_1~=~\sum_{j=1}^{d}|w_{j}-w_{d+j}|~\leq~\sum_{j=1}^{d}\left(|w_{j}|+|w_{d+j}|\right)~=~1~.
    \]
    Moreover, it is easily verified that $A\bw=\psi(A)\psi(\bw)$ and   $A\psi^{-1}(\bw)=\psi(A)\bw$. 
    
    Now, given the algorithm $\Acal$ for the simplex domain and a matrix $A\in\reals^{n\times d}$, consider the following algorithm for the $\ell_1$ domain: Run $\Acal$ on the matrix $\psi(A)\in[-1,+1]^{n\times 2d}$ for $T$ iterations, resulting in the vector $\bw_{T+1}\in \Delta^{2d-1}$, and return the vector $\psi^{-1}(\bw_{T+1})\in\reals^{d}$ (which has $\ell_1$ norm at most $1$ by the above). Letting $\bw^*$ be some optimal solution of $\max_{\bw\in\reals^d:\norm{\bw}_1\leq 1}\min_{\bp\in\Delta^{n-1}}\bp^\top A\bw$, we have
    \begin{align*}
    &\left(\max_{\bw\in\reals^d:\norm{\bw}_1\leq 1}\min_{\bp\in\Delta^{n-1}}\bp^\top A\bw\right)-\left(\min_{\bp\in\Delta^{n-1}}\bp^\top A\psi^{-1}(\bw_{T+1})\right)\\
    &=~
\left(\min_{\bp\in\Delta^{n-1}}\bp^\top A\bw^*\right)-\left(\min_{\bp\in\Delta^{n-1}}\bp^\top A\psi^{-1}(\bw_{T+1})\right)\\
&=~
\left(\min_{\bp\in\Delta^{n-1}}\bp^\top \psi(A)\psi(\bw^*)\right)-\left(\min_{\bp\in\Delta^{n-1}}\bp^\top \psi(A)\bw_{T+1}\right)\\
&\leq~
\left(\max_{\bw\in\Delta^{2d-1}}\min_{\bp\in\Delta^{n-1}}\bp^\top \psi(A)\bw\right)-\left(\min_{\bp\in\Delta^{n-1}}\bp^\top \psi(A)\bw_{T+1}\right)
~\leq~ \epsilon(T,n,d)
\end{align*}
as required, where the last inequality is by the guarantee on the algorithm $\Acal$.

\subsection{Auxiliary lemmas}

\begin{lemma}\label{lem:rescaledGaussian}
    Let $\bxi\sim \Ncal(\mathbf{0},I_q)$ be a standard Gaussian random variable,
    and define the vector $\bx=\beta(I-MM^\top)\bxi$ where $M$ is some matrix composed of orthonormal columns, and $\beta>0$. Then for all $z>0$,
    \[
    \Pr(\norm{\bx}_{\infty}\geq z)~\leq~ 2q\exp\left(-\frac{z^2}{2\beta^2}\right)~.
    \]
    Moreover, if the number of columns in $M$ is less than $q/2$, then it also holds that
    \[
    \Pr\left(\frac{\norm{\bx}_{\infty}}{\norm{\bx}_2}\geq \frac{z}{\sqrt{q}}\right)~\leq~ 2q\exp\left(-\frac{z^2}{8}\right)+\exp\left(-\frac{q}{48}\right)~.    
    \]
\end{lemma}
\begin{proof}[Proof of \lemref{lem:rescaledGaussian}]
 $\bx$ has a zero-mean Gaussian distribution, with covariance matrix $\beta^2(I-M M^\top)^2=\beta^2(I-M M^\top)$. In particular, for any fixed index $l$, the coordinate $x_{l}$ is a zero-mean univariate Gaussian with variance $\beta^2\be_{l}^\top (I-M M^\top)\be_{l}=\beta^2(1-\norm{\be_{l}^\top M}^2)\leq \beta^2$. Therefore, by a standard Gaussian tail bound,
$\Pr\left(|x_{l}|\geq z\right)\leq2\exp\left(-{z^2}/{2\beta^2}\right)$ for any $z>0$. Applying a union bound over all indices $l\in [q]$, we get
    \begin{equation}\label{eq:infbound}
    \Pr\left(\norm{\bx}_{\infty}\geq z\right)~\leq~ 2q\exp\left(-\frac{z^2}{2\beta^2}\right)~,
    \end{equation}
    from which the first inequality in the lemma follows. 

    As for the second inequality, note that since $M$ has at most $q/2$ orthonormal columns, then $I-MM^\top$ is a projection matrix to a subspace of $\reals^q$ of dimensionality at least $q/2$. Thus, $\norm{(I-MM^\top)\bxi}$ has the same distribution as the norm of a standard Gaussian random variable on $\reals^s$ where $s\geq q/2$. Using a standard tail lower bound for such Gaussian norms (see for example \cite[Lemma B.12]{shalev2014understanding}), it follows that
\begin{align*}
\Pr\left(\norm{(I-MM^\top)\bxi}^2\leq \frac{1}{2}\cdot \frac{q}{2}\right)
&\leq
\Pr\left(\norm{(I-MM^\top)\bxi}^2\leq \frac{1}{2}\cdot s\right)
\leq\exp\left(-\frac{s}{24}\right)
\\&\leq\exp\left(-\frac{q}{48}\right)~.
\end{align*}   
    Therefore, since $\bx=\beta(I-MM^\top)\bxi$, 
    \[
    \Pr\left(\norm{\bx}_2\leq \frac{\beta\sqrt{q}}{2}\right)~=~
    \Pr\left(\norm{(I-MM^\top)\bxi}^2\leq \frac{q}{4}\right)~\leq~\exp\left(-\frac{q}{48}\right)~.
    \]
    Combining this with \eqref{eq:infbound} using a union bound, it follows that
    \[
    \Pr\left(\frac{\norm{\bx}_{\infty}}{\norm{\bx}_2}\geq \frac{2z}{\beta\sqrt{q}}\right)~\leq~ 2q\exp\left(-\frac{z^2}{2\beta^2}\right)+\exp\left(-\frac{q}{48}\right)~.    
    \]
    Substituting $z$ instead of $2z/\beta$ and simplifying a bit, the result follows.
\end{proof}

\begin{lemma} \label{lem: exp tail bound}
Let $\xi_1,\dots,\xi_T$ be a sequence of random variables such that for all $t\in[T]:~\xi_t\mid \xi_{1},\dots,\xi_{t-1}\sim \Ncal(0,\sigma_t^2)$, and further assume that $\sigma_t^2\leq\beta^2$ for all $t$. Then
with probability at least $1-\delta:~\sum_{t=1}^{T}\xi_t^2\leq 8\beta^2(T+\log(1/\delta))$.
\end{lemma}

\begin{proof}[Proof of \lemref{lem: exp tail bound}]
Denote $S:=\sum_{t=1}^{T}\xi_t^2,~V:=\E[S]=\sum_{t=1}^{T}\sigma_t^2\leq \beta^2 T$ (by the law of total expectation). Fix $\lambda\in[0,1/4\beta^2]$ and let
\[
M_t:=\exp(\lambda\sum_{i=1}^{t}\xi_i^2)\prod_{i=1}^{t}\sqrt{1-2\lambda\sigma_i^2}~,
\]
where we also let $M_0:=1$. Then
\begin{align*}
\E[M_t\mid \xi_1,\dots,\xi_{t-1}]=\exp(\lambda\sum_{i=1}^{t}\xi_i^2)\prod_{i=1}^{t}\sqrt{1-2\lambda\sigma_i^2}\cdot \E[e^{\lambda\xi_t^2}\mid \xi_1,\dots,\xi_{t-1}]\sqrt{1-2\lambda\sigma_t^2}=M_{t-1}~,
\end{align*}
thus $(M_t)_{t=0}^{T}$ is a martingale, and in satisfies $\E[M_T]=1$. By Markov's inequality, for any $z>0$,
\begin{align}
\Pr(S\geq z)=\Pr(e^{\lambda S}\geq e^{\lambda z})\leq e^{-\lambda z}\E[e^{\lambda S}]&=e^{-\lambda z}\underset{=\E[M_T]=1}{\underbrace{\E\left[e^{\lambda S}\prod_{t=1}^{T}\sqrt{1-2\lambda\sigma_t^2}\right]}}\cdot\left(\prod_{t=1}^{T}\sqrt{1-2\lambda\sigma_t^2}\right)^{-1}
\nonumber\\&=
\exp\left(-\lambda z-\frac{1}{2}\sum_{t=1}^{T}\log(1-2\lambda\sigma_t^2)\right)~.\label{eq: markov}
\end{align}
To simplify the bound, we note that $\lambda\leq 1/4\beta^2$ ensures that $2\lambda\sigma_t^2\leq\half$, so we can use the numerical inequality $-\log(1-a)\leq 2a$ (which holds for $a\in[0,\half]$) and get that $-\half\sum_{t=1}^{T}\log(1-2\lambda\sigma_t^2)\leq 2\lambda V$. Plugged into \eqref{eq: markov}, after a change of variable we see that
\[
\Pr(S\geq V+z)\leq \exp(-\lambda z+\lambda V)~.
\]
Concretely, setting $\lambda=1/8\beta^2$ and $z=\beta^2T+8\beta^2\log(1/\delta)$ yields
\begin{align*}
\Pr\left(\sum_{t=1}^{T}\xi_t^2\geq 8\beta^2(T+\log(1/\delta))\right)&\leq
\Pr\left(S\geq V+\beta^2T+8\beta^2\log(1/\delta))\right)
\\&\leq \exp\left(-\frac{T}{8}-\log(1/\delta)+\frac{V}{8\beta^2}\right)
\leq\delta~.
\end{align*}

\end{proof}

\section{Discussion}\label{sec:discussion}

In this paper, we studied the oracle complexity of solving matrix games where (at least) one of the constraint sets is the simplex. This well-studied problem class encapsulates canonical tasks, including finding a linear separator and computing a Nash equilibrium in bilinear zero-sum games, as well as some others discussed throughout the paper.
By identifying distinct oracle models corresponding to either one-sided or two-sided multiplication by the matrix $A$, we were able to shed light on previous algorithmic approaches for this task, and to provide several new lower bounds. 

The most notable implications of our work are that alongside the recent upper bounds of \citet{karmarkar2026solving}, our work resolves the complexity of these two canonical tasks on the order of $\tilde{\Theta}(\epsilon^{-2/3})$ (up to log factors), in the natural oracle model in which the interaction with the matrix is granted through multiplications from both sides. Moreover, our lower bound in the one-sided multiplication model points out that any acceleration beyond the classic $O(\epsilon^{-2})$ rate provably requires such two-sided operations.

We note that extending our lower bounds to randomized algorithms remains an open problem. Since our constructions required simulation of the algorithm's responses, it is not immediately clear how this can be achieved. We further note that previous extensions of oracle complexity lower bounds to randomized algorithms \citep{nemirovskiyudin1983,carmon2020lower,arjevani2023lower} crucially relied on non-linear modifications of the ``hard'' target function, which is not possible when we restrict ourselves to bilinear functions.

\subsection*{Acknowledgments}
This research is supported in part by European Research Council (ERC) grant 754705. GK is supported by an Azrieli Foundation graduate fellowship.

\bibliographystyle{plainnat}
\bibliography{mybib}

\end{document}